\documentclass[12pt]{article}
\usepackage{fullpage}
\usepackage[T1]{fontenc}
\usepackage{ae,aecompl}
\usepackage[dvips]{graphicx}
\usepackage{amssymb}
\usepackage{amsmath}
\usepackage{subfigure}
\usepackage[round,authoryear]{natbib}
\usepackage{color}
\usepackage{array}
\usepackage{rotating}
\usepackage{colortbl}
\usepackage{tikz}
\usepackage{url}
\usepackage[pdfborder=false]{hyperref}
\usepackage{sectsty}

\usetikzlibrary{patterns}
\definecolor{webgreen}{rgb}{0,0.4,0}
\definecolor{webbrown}{rgb}{0.6,0,0}
\definecolor{purple}{rgb}{0.5,0,0.25}
\definecolor{darkblue}{rgb}{0,0,0.7}
\definecolor{darkred}{rgb}{0.7,0,0}
\definecolor{darkgreen}{rgb}{0,0.7,0}
\hypersetup{colorlinks,citecolor=darkred,filecolor=black,linkcolor=darkblue,urlcolor=webgreen,pdfpagemode=none,
pdfstartview=FitH}
\newcommand{\ignore}[1]{}
\newtheorem{lemma}{{\sc Lemma}}
\newtheorem{prop}{{\sc Proposition}}
\newtheorem{cor}{{\sc Corollary}}
\newtheorem{theorem}{{\sc Theorem}}
\newtheorem{defn}{{\sc Definition}}

\newtheorem{fact}{{\sc Fact}}

\newenvironment{proof}{\noindent {\bf \sl Proof\/}:\enspace}
{\hfill $\blacksquare{}$ \vspace{12pt}}
\sectionfont{\centering\normalfont\scshape}
\subsectionfont{\centering\normalfont}
\begin{document}

\title{\textbf{\Large Balanced Ranking Mechanisms}~\thanks{
We are grateful to Dilip Abreu, Attila Ambrus, Claude d'Aspremont, Sushil Bikhchandani,
Rahul Deb, Bhaskar Dutta, Ed Green, Matt Jackson, Eric Maskin, Motty Perrry, Arunava
Sen, Ricard Torres, Levent Ulku, Dan Vincent and seminar participants at
ITAM, ISI Delhi, and Delhi School of Economics for useful comments.}}
\author{Debasis Mishra and Tridib Sharma~\thanks{%
Debasis Mishra: Indian Statistical Institute, Delhi, \texttt{%
dmishra@isid.ac.in; dmishra@gmail.com}; Tridib Sharma: ITAM, Mexico, \texttt{%
sharma@itam.mx}}}
\date{{\small {\today}}}
\maketitle

\begin{abstract}
In the private values single object auction model, we construct a \emph{%
satisfactory} mechanism - a symmetric, dominant strategy incentive
compatible, and budget-balanced mechanism. Our mechanism allocates the
object to the highest valued agent with more than 99\% probability provided
there are at least 14 agents. It is also ex-post individually rational. We
show that our mechanism is optimal in a restricted class of satisfactory 
\emph{ranking} mechanisms. Since achieving efficiency through a dominant
strategy incentive compatible and budget-balanced mechanism is impossible in
this model, our results illustrate the limits of this impossibility. \newline
\newline

\noindent \textsc{Keywords.} budget-balanced mechanisms, Green-Laffont
mechanism, Pareto optimal mechanism. \newline

\noindent \textsc{JEL Keywords.} D82, D71, D02.
\end{abstract}

\newpage

\section{Introduction}

How should a group of agents allocate a unit of resource among themselves?
For instance, consider the problem of allocating a bequest among a group of
potential heirs. Many a times, no will exists. 
Even when a will exists, disputes arise.
Designated estate agents are often employed to resolve bequest related
problems. A Wall Street Journal article quotes an expert suggesting the
following dispute resolution procedure:

\begin{quote}
\textsl{In family disputes, Ms. Olsavsky says, one option is to have all the
items put up for auction. Family members can bid on what they want. The
money goes back to the estate to be divided equally \citep{wsj}.}
\end{quote}

There are a number of other examples: a group of firms sharing time slots on
a jointly owned supercomputer \citep{Guo11}; a
group of municipalities deciding on the location of a stadium \citep{Cramton87}. A key feature
of these problems is that transfers can be used (either as taxes or
subsidies) for resource allocation. However, transfers across agents have to
balance - money raised by auctioning a bequest must be redistributed among
the heirs.~\footnote{Commenting on the recent controversy surrounding
the allocation of soccer World cup venue, Rakesh Vohra in the popular blog \emph{The Leisure of the Theory Class} writes:
\textsl{Instead of running beauty contests to decide where to hold FIFA
events, auction off the right to the highest bidder. This can be done in two
ways. Allow each FIFA official with a vote to auction off their vote to the
highest bidder. Or, do away with the officials altogether and have countries
bid directly for the right to hold FIFA events. Full transparency, no
bribery and FIFA may be richer than before!~\citep{ltc}}
Just like the bequest settlement case, FIFA must redistribute the transfers
collected from the countries among them.}

We design mechanisms for such problems with the aim of achieving efficiency.
Efficiency requires one to allocate the bequest to the highest valued heir
or to allocate the world cup venue to the country which benifits the most
from hosting the event. In the standard private values model, where each
agent has a value for the unit of resource/object and transfers are allowed
with quasilinear utility, the Vickrey auction satisfies three compelling
desiderata of a mechanism: (a) dominant strategy incentive compatibility
(DSIC), (b) (allocative) efficiency - allocating the object to the highest
valuation agent, and (c) ex-post individual rationality. A well-known
criticism of the Vickrey auction is that it is not budget-balanced - it
collects revenue from the agents, which distorts ex-post efficiency. \cite%
{Green79} shows that this criticism applies to every DSIC and efficient
mechanism: no DSIC and efficient mechanism can be budget-balanced. We look
for a second-best solution, where we explore the limits of this
impossibility result:

\begin{quote}
\textsl{How close to efficiency can we get using a DSIC and budget-balanced
mechanism?}
\end{quote}

We require our solution to satisfy symmetry - agents with identical
valuation must get the object with equal probability and pay the same
amount. Symmetry is a compelling fairness property - for instance, in the
bequest allocation problem, an asymmetric mechanism may either be
unacceptable to potential heirs or lead to unpleasant lawsuits later on.

We identify a class of DSIC, budget-balanced, and symmetric mechanisms that
we call \emph{ranking mechanisms}. A ranking mechanism is one that uses a 
\emph{ranking} allocation rule, which is specified (for $n$ agents) by $n$
numbers $(\pi _{1},\ldots ,\pi _{n})$ between $0$ and $1$ such that they add
up to not more than $1$ and $\pi _{j}\geq \pi _{j+1}$ for each $j$. For
every $j$, the number $\pi _{j}$ is the probability with which an agent with
the $j$-th highest value is allocated the object at any generic profile of
values. Our main result is a description of the \emph{r-optimal} mechanism -
a DSIC, budget-balanced, and symmetric ranking mechanism that beats every
such mechanism in terms of the allocation probability to the highest
valuation agent.

At every profile of values, our r-optimal mechanism allocates the object to
the highest valued agent with more than 99\% probability, provided there are
at least 14 agents. It is also ex-post individually rational. The welfare
generated by the r-optimal mechanism converges to efficiency as the number
of agents increase. The nature of convergence is shown in Table \ref{tab:eff}%
, where we report on the probability with which the highest valued agent
gets the object in our mechanism. 
\begin{table}[!htb]
\centering
\begin{tabular}{c|c}
No of agents & Probability to the highest valued agent \\ \hline
$9$ & $92.3\%$ \\ 
$10$ & $95\%$ \\ 
$11$ & $96.2\%$ \\ 
$12$ & $98.1\%$ \\ 
$13$ & $98.9\%$ \\ 
$14$ & $99.4\%$ \\ 
$15$ & $99.6\%$ \\ 
$16$ & $99.8\%$ \\ 
$17$ & $99.9\%$%
\end{tabular}%
\caption{Convergence in our mechanism}
\label{tab:eff}
\end{table}

The r-optimal mechanism we identify satisfies ex-post individual rationality.
Ex-post individual rationality is a desired property of mechanisms.
Consider politicians across municipalities or countries, involved in
procuring a public facility or negotiating the venue of an international
sporting event. Failure to get the facility or the event would result in
criticisms. The criticism would be compounded if in addition to a failure,
payments also have to be made. Political opponents could well allege
corruption.

Ranking mechanisms contain two familiar DSIC, budget-balanced, and symmetric
mechanisms: (i) the mechanism that allocates the object to each agent with
equal probability without using any transfers and (ii) the \emph{residual
claimant} mechanism in \citet{Green79}. The residual claimant mechanism is
defined by choosing an agent uniformly at random as a residual claimant and
conducting a Vickrey auction among the other agents. The revenue generated
from the auction is then given to the residual claimant. We refer to this
mechanism as the Green-Laffont (GL) mechanism, and note that at profiles of
distinct values, it allocates the object to the highest valued agent with
probability $1-1/n$ and to the second highest valued agent with probability $%
1/n$.~\footnote{%
This mechanism (and its variants) were discussed in the context of
public-good provision problem in \citet{Green79}. Later, \citet{Gary00}
formally define this mechanism and study its statistical and strategic
properties.} Our r-optimal mechanism coincides with the GL mechanism if the
number of agents is no more than $8$ but differs from it significantly for
more than $8$ agents.

Our analysis is prior-free. We use DSIC as our solution concept. As we
discuss later in Section \ref{sec:lit}, \citet{Cramton87} show that Bayesian
incentive compatible, efficient, and budget-balanced mechanisms satisfying
a form of individual rationality exists in our model. While the mechanism they
propose require information about beliefs of agents (with common prior assumption), our result shows the level of efficiency that can be achieved using DSIC and budget-balanced mechanisms, thus showing
the limits of such a prior-free and robust approach in this problem.
Inspired by the seminal work of \citet{Bergemann05}, recent literature in
mechanism design has been investigating such questions in other models~%
\citep{Chung07,Carroll15}.~\footnote{%
There are two recent papers which also provide foundational results of DSIC
mechanisms in the private values single object auction environment. %
\citet{Manelli10} show that in such models, for every Bayesian incentive
compatible mechanism, there is an \textquotedblleft equivalent" DSIC
mechanism - this equivalence is in terms of interim expected utility of
agents. This result is extended to other settings in \citet{Gershkov13}.
Unlike our work, these papers do not impose budget-balance as a constraint -
indeed, these equivalence results do not hold if budget-balance constraint
is imposed.}

In view of the Green and Laffont impossibility result, comparing efficiency
levels of two DSIC, budget-balanced, and symmetric mechanisms is a natural
question. The notion we use here compares ranking mechanisms by the
probability with which the highest valued agent gets the object. Formally,
we show that this notion coincides with a \emph{worst-case} measure of
efficiency: the worst-case ratio of welfare generated by a ranking mechanism
and efficient level of welfare. In a prior-free environment, such worst-case
measures give a very robust method of comparing mechanisms. These measures
are widely used to compare algorithms in the computer science literature,
and in the algorithmic game theory literature~\citep{Cavallo06,Guo09}. They
are also becoming popular in the mechanism design literature~%
\citep{Chung07,Moulin09,Carroll15,Masso15}.

From a technical point, our paper extends the Myersonian approach. Recall that \citet{Myerson81} provides
necessary and sufficient conditions for a mechanism to be DSIC.~\footnote{His
characterization is for Bayesian incentive compatible mechanisms, but
can be straightforwardly adapted to DSIC mechanisms.} We extend his characterization
to give necessary and sufficient conditions for a mechanism to be DSIC, budget-balanced,
and symmetric. One of the surprising corollaries of this characterization is that if there is a DSIC, budget-balanced,
and symmetric mechanism using an allocation rule, then it is the only such mechanism
using this allocation rule. A consequence of this result is that the search 
over the domain of DSIC, budget-balanced, and symmetric mechanisms can be confined to the domain
of allocation rules satisfying our necessary and sufficient conditions - we do not have
to worry about payments since they are identified uniquely. 
Our characterization
reveals a rich but complex class of such mechanisms. The ranking mechanisms
that we consider in this paper are much simpler to describe. 
The separation of payment and allocation decisions gives
us a lot of tractability in the class of ranking allocation rules, where we derive our
mechanism and show its constrained optimality. Though we do
not know if we can improve upon our r-optimal mechanism, by considering more
complex mechanisms, the overwhelming speed of convergence of our mechanism
(as shown in Table \ref{tab:eff}) implies that we may not be losing out much
by restricting attention to ranking mechanisms.

The rest of the paper is organized as follows. We present our model in
Section \ref{sec:mod}. We introduce ranking mechanisms and discuss our main
results in Section \ref{sec:rank}. We give a technical characterization of
DSIC, budget-balanced, and symmetric mechanisms in Section \ref{sec:proof}.
We relate our results to the literature in Section \ref{sec:lit} and
conclude in Section \ref{sec:conc}. All the omitted proofs are relegated to
an Appendix at the end. To keep the proofs of our results lucid, we present
them in a different sequence than the sequence in which corresponding
results appear in the main text. Hence, we recommend that the proofs be read
after reading the main text.

\section{The Model}

\label{sec:mod}

We consider the standard single object independent private values model with 
$N=\{1,\ldots,n\}$ as the set of agents. Throughout, we assume that $n \ge 3$
- the $n=1$ case is trivial and the $n=2$ case is discussed later. Each
agent $i \in N$ has a valuation $v_i$ for the object. If he is given $%
\alpha_i \in [0,1]$ of the object, or given the object with probability $%
\alpha_i$, and he pays $p_i$ for it, then his net utility is $\alpha_i v_i -
p_i$. The set of all valuations for any agent is given by $V \equiv
[0,\beta] $, where $\beta \in \mathbb{R}$. A valuation profile will be
denoted by $\mathbf{v} \equiv (v_1,\ldots,v_n)$.

An \textbf{allocation rule} is a map $f: V^n \rightarrow [0,1]^n$, where we
denote by $f_i(\mathbf{v})$ the probability of agent $i$ getting allocated
the object at valuation profile $\mathbf{v}$. We assume that at all $\mathbf{%
v} \in V^n$, $\sum_{i \in N}f_i(\mathbf{v}) \le 1$.

A \textbf{payment rule} of agent $i$ is a map $p_i:V^n \rightarrow \mathbb{R}
$. A collection of payment rules of all the agents will be denoted by $%
\mathbf{p} \equiv (p_1,\ldots,p_n)$. A \textbf{mechanism} is a pair $(f,%
\mathbf{p})$. We require our mechanism to satisfy the following three
properties:

\begin{itemize}
\item A mechanism $(f,\mathbf{p})$ is \textbf{dominant strategy incentive
compatible (DSIC)} if for every $i \in N$, for every $v_{-i} \in V^n$, and
for every $v_i,v^{\prime }_i \in V$, we have 
\begin{equation*}
v_if_i(v_i,v_{-i}) - p_i(v_i,v_{-i}) \ge v_if_i(v^{\prime }_i,v_{-i}) -
p_i(v^{\prime }_i,v_{-i}).
\end{equation*}

\item A mechanism $(f,\mathbf{p})$ is \textbf{budget-balanced (BB)} if for
every $\mathbf{v} \in V^n$, we have 
\begin{equation*}
\sum_{i \in N}p_i(\mathbf{v}) = 0.
\end{equation*}

\item A mechanism $(f,\mathbf{p})$ is \textbf{symmetric} if for every $%
\mathbf{v} \in V^n$ and for every $i, j \in N$ with $v_i = v_j$, we have 
\begin{equation*}
f_i(\mathbf{v})=f_j(\mathbf{v}),~~~~p_i(\mathbf{v})=p_j(\mathbf{v}).
\end{equation*}
\end{itemize}

We call a mechanism \textbf{satisfactory} if it is DSIC, BB, and symmetric.~\footnote{\citet{Green77} use the terminology {\em satisfactory} mechanism to
mean something different. Among other things, 
their satisfactory mechanisms are DSIC and {\em efficient} but need not be BB and symmetric.
We apologize if this creates a confusion.
}
Symmetry allows us to consider a mild notion of fairness in our mechanism.
It also explicitly rules out \emph{dictatorial} mechanisms, where a dictator
agent is given the object for free at all valuation profiles.~\footnote{%
A weaker version of symmetry would be to consider \emph{anonymity} of the
mechanism with respect to net utilities of the agents - see %
\citet{Sprumont13} for a formal definition. We will require our stronger
version of symmetry for our main result.}

An allocation rule $f$ is \textbf{satisfactorily implementable} if there
exists a $\mathbf{p}$ such that $(f,\mathbf{p})$ is a satisfactory
mechanism. We are interested in finding satisfactory mechanisms that are
almost efficient in the following sense.

At any valuation profile $\mathbf{v}$, denote by $\mathbf{v}[k]$ the set of
agents who have the $k$-th highest valuation at $\mathbf{v}$. More formally, 
\begin{equation*}
\mathbf{v}[1]:=\{i \in N: v_i \ge v_j~\forall~j \in N\}.
\end{equation*}
Having defined $\mathbf{v}[k-1]$, we recursively define $\mathbf{v}[k]$ as 
\begin{equation*}
\mathbf{v}[k]:=\{i \in N \setminus (\cup_{k^{\prime }=1}^{k-1}\mathbf{v}%
[k^{\prime }]): v_i \ge v_j~\forall~j \in N \setminus (\cup_{k^{\prime
}=1}^{k-1}\mathbf{v}[k^{\prime }])\}.
\end{equation*}

\begin{defn}
An allocation rule $f$ is \textbf{efficient} at $\mathbf{v}$ if 
\begin{equation*}
\sum_{i \in \mathbf{v}[1]}f_i(\mathbf{v})=1.
\end{equation*}
An allocation rule $f$ is efficient if it is efficient at all $\mathbf{v}
\in V^n$. A mechanism $(f,\mathbf{p})$ is efficient if $f$ is efficient.
\end{defn}

The efficiency of a BB mechanism is equivalent to maximizing the total
welfare of agents at every profile of valuations. To see this, note that the
total welfare of agents at a valuation profile $\mathbf{v}$ from a mechanism 
$(f,\mathbf{p})$ is 
\begin{equation*}
\sum_{i \in N}\Big[ v_i f_i(\mathbf{v}) - p_i(\mathbf{v}) \Big] = \sum_{i
\in N}v_if_i(\mathbf{v}),
\end{equation*}
where the second equality followed from BB. This is clearly maximized by
assigning the object to the highest valued agents.

\citet{Green79} show that no DSIC and budget-balanced mechanism can be
efficient. Hence, a satisfactory mechanism cannot be efficient. The precise
question we are interested in is: \emph{what is the ``most" efficient
satisfactory mechanism?}

\subsection{A Prior-Free Notion to Measure Efficiency}

In view of the Green-Laffont result, we adopt one of the well-known notions
to measure efficiency of satisfactory mechanisms. Fix a satisfactory
mechanism $\mathcal{M} \equiv (f,\mathbf{p})$. Note that at any valuation
profile $\mathbf{v}$ with $v_1 \ge \ldots \ge v_n$, the maximum possible
(efficient) social welfare is $v_1$, and the social welfare achieved by $%
\mathcal{M}$ is 
\begin{equation*}
\sum_{i \in N}v_i f_i(\mathbf{v}).
\end{equation*}
The ratio of these two numbers is a good measure of efficiency at the
valuation profile $\mathbf{v}$. More precisely, the number 
\begin{equation*}
f_1(\mathbf{v}) + \frac{1}{v_1}\Big( \sum_{i \ne 1} v_i f_i(\mathbf{v}))\Big)%
,
\end{equation*}
is a measure of efficiency at the valuation profile $\mathbf{v}$. Here, as
in the rest of the paper, we assume $\frac{0}{0}=1$. Note that such a
measure \emph{only} depends on $f$ and not on $\mathbf{p}$ because $(f,%
\mathbf{p})$ is a budget-balanced mechanism. Now, the \emph{worst-case} of
this ratio happens when we minimize this over all $\mathbf{v}$. In
particular, for a satisfactory mechanism $\mathcal{M} \equiv (f,\mathbf{p})$%
, the worst-case efficiency is given by

\begin{equation*}
\mu^{\mathcal{M}} = \inf_{\mathbf{v}} \Bigg[ f_1(\mathbf{v}) + \frac{1}{v_1}%
\Big( \sum_{i \ne 1} v_i f_i(\mathbf{v}))\Big) \Bigg].
\end{equation*}

A natural objective is to find a satisfactory mechanism that maximizes this
worst-case efficiency. As discussed in the introduction, this is a robust
method of comparing efficiency of mechanisms. We apply this notion of
comparing efficiency levels of mechanisms in a restricted class of
mechanisms that we describe next.

\section{Ranking Mechanisms}

\label{sec:rank}

In most of the paper, we focus attention on the following class of simple
allocation rules and the corresponding satisfactory mechanisms that can be constructed using
such allocation rules. We call them {\em ranking} allocation rules.
Each ranking allocation rule is defined by $n$ numbers
$(\pi_1,\ldots,\pi_n)$ with each $\pi_i \in [0,1]$ and $\sum_{i \in N}\pi_i \le 1$.
Informally, at a generic valuation profile $v_1 > v_2 > \ldots > v_n$, for every $k$, $\pi_k$ reflects the probability with which agent $k$ (which has rank $k$ at this profile)
gets the object.
Notice that this probability does not change across valuation profiles as long as the rank 
of the agent does not change.
This feature makes the ranking allocation rules simple, both from the point of view
of practical implementation and analysis.
Also, we require every ranking
allocation rule to be symmetric, and this means that it allocates the object in a
particular way when there are ties in valuations. We clarify this tie-breaking by formally defining
the ranking allocation rule first.
\begin{defn}
\label{def:rank} An allocation rule $f$ is a ranking allocation rule if it
is symmetric and there exists numbers $\pi_i \in [0,1]$ for all $i \in N$
with $\pi_1 \ge \ldots \ge \pi_n$ and $\sum_{i \in N}\pi_i \le 1$ such that
at every valuation profile $\mathbf{v}$ and every $k \in N$, we have 
\begin{equation*}
\sum_{i \in \cup_{j=1}^kv[j]}f_i(\mathbf{v}) = \sum_{i \in
\cup_{j=1}^kv[j]}\pi_i.
\end{equation*}
A mechanism $(f,\mathbf{p})$ is a ranking mechanism if $f$ is a ranking
allocation rule.
\end{defn}
To illustrate the tie-breaking, suppose there are seven
agents: $N=\{1,\ldots,7\}$ and consider a valuation profile $\mathbf{v}$
such that $v_1=v_2 > v_3=v_4=v_5 > v_6 > v_7$. Consider a ranking allocation
rule $(\pi_1,\ldots,\pi_7)$. According to the definition, agents $1$ and $2$
will equally share (due to symmetry) the allocation probabilities $%
(\pi_1+\pi_2)$, i.e., each agent gets the good with probability $\frac{%
\pi_1+\pi_2}{2}$. Then, agents $3,4,$ and $5$ will equally share the
allocation probabilities $(\pi_3+\pi_4+\pi_5)$. Finally, agents $6$ and $7$
get allocation probabilities $\pi_6$ and $\pi_7$ respectively.

Note that breaking ties in this manner in a ranking allocation rule
maintains continuity of total welfare in terms of valuations of agents. For
instance, consider the valuation profile discussed in the above example.
Consider any arbitrarily close \emph{generic} (with distinct valuations for
agents) valuation profile to this valuation profile. The total expected
value of agents $1$ and $2$ in this profile is arbitrarily close to $%
v_1\pi_1+v_2\pi_2=v_1(\pi_1+\pi_2)$, where the equality follows from the
fact that $v_1=v_2$. Hence, we can maintain continuity of total welfare by
giving a total of $(\pi_1+\pi_2)$ probability to agents $1$ and $2$.
Finally, using symmetry, we equally divide this probability among these two
agents. This explains the tie-breaking in the ranking allocation rule.

Even though the ranking allocation rule is a simple class of allocation
rules, there is a rich subclass of ranking allocation rules that are
satisfactorily implementable. Our focus on this class is purely driven by
their tractability and simplicity.

Two well-known ranking allocation rules are satisfactorily implementable.
The equal-sharing allocation rule, where each agent gets the object with
probability $\frac{1}{n}$ is satisfactorily implementable - no transfers are
required for this. The other allocation rule comes from a mechanism proposed
by Green and Laffont. Pick an agent $i$ uniformly at random. Run a Vickrey
auction among the remaining $N \setminus \{i\}$ agents. Give the revenue
from the Vickrey auction to agent $i$. Since agents are treated
symmetrically, the Vickrey auction is DSIC, and by construction, the
mechanism is budget-balanced.~\footnote{\citet{Green79} discuss an even
larger class of satisfactory mechanisms where they take out a coalition of
``residual claimant" agents with some probability, run the Vickrey auction
on the remaining agents, and allocate the revenue of the Vickrey auction to
the residual claimants equally. These mechanisms are also ranking mechanisms.%
}

A closer look at the Green-Laffont mechanism reveals the following. For
valuation profiles with a distinct highest valued agent, it allocates the
object to him with probability $(1-1/n)$ and shares the remaining
probability $1/n$ among the second highest valued agents. For valuation
profiles with more than one highest valued agents, it allocates the entire
object equally among the highest valued agents. Therefore, given Definition %
\ref{def:rank}, the allocation rule used in the Green-Laffont mechanism is a
ranking allocation rule, where 
\begin{equation*}
\pi_1=1-1/n, \pi_2=1/n, \pi_3=\ldots=\pi_n=0.
\end{equation*}
To be precise, this is the allocation rule corresponding to the direct
mechanism of the Green-Laffont mechanism.

We now characterize the ranking allocation rules that can be satisfactorily
implemented. \newline

\noindent \textbf{Notation:} For any two non-negative numbers $K$ and $%
K^{\prime }$ with $K \ge K^{\prime }$, we denote by $C(K,K^{\prime })$ the
number of ways we can choose $K^{\prime }$ agents from a set of $K$ agents.

\begin{prop}
\label{prop:rank} A ranking allocation rule with probabilities $%
(\pi_1,\ldots,\pi_n)$ is satisfactorily implementable if and only if 
\begin{equation*}
\sum_{k=1}^n (-1)^k C(n-1,k-1)\pi_k = 0.
\end{equation*}
\end{prop}

Later, in Theorem \ref{theo:char}, we give necessary and sufficient
conditions for a general allocation rule $f$ to be satisfactorily
implementable. Those necessary and sufficient conditions are complicated -
they involve verifying an infinite system of equations. On the other hand,
the necessary and sufficient condition for satisfactorily implementing a
ranking allocation rule is a single equation given by Proposition \ref%
{prop:rank}. This hints that it may be tractable to search over the space of
ranking allocation rules.

Now, we adapt our notion of efficiency measure by restricting the class of
mechanisms to ranking mechanisms.

\begin{defn}
A ranking allocation rule $(\pi_1,\ldots,\pi_n)$ is \textbf{r-optimal} if it
satisfactorily implementable and for any other satisfactorily implementable
ranking allocation rule $(\pi_1^{\prime },\ldots,\pi^{\prime }_n)$, we have 
\begin{equation*}
\pi_1 \ge \pi^{\prime }_1.
\end{equation*}

A ranking mechanism $(f,\mathbf{p})$ is r-optimal if (i) $(f,\mathbf{p})$ is
a satisfactory mechanism and (ii) $f$ is r-optimal.
\end{defn}

The notion of r-optimality is an indirect way of requiring a mechanism to
maximize the value of worst-case efficiency in the class of satisfactory
ranking mechanisms. To see this, fix a ranking mechanism $\mathcal{M}\equiv
(f,\mathbf{p})$ with allocation probabilities $(\pi _{1},\ldots ,\pi _{n})$.
Note that 
\begin{equation*}
\mu ^{\mathcal{M}}=\inf_{\mathbf{v}}\Big[\pi _{1}+\frac{1}{v_{1}}\big(%
\sum_{j\neq 1}\pi _{j}v_{j}\big)\Big]=\pi _{1}+\inf_{\mathbf{v}}\frac{1}{%
v_{1}}\big(\sum_{j\neq 1}\pi _{j}v_{j}\big)=\pi _{1},
\end{equation*}%
where we used the fact that infimum of the above expression occurs when each
agent $j\neq 1$ has zero valuation.

Later, in Theorem \ref{theo:char}, we shall establish the fact that if $f$
is satisfactorily implementable, then there is a unique $\mathbf{p}$ such
that $(f,\mathbf{p})$ is a satisfactory mechanism. As a result, we shall
only talk about the r-optimality of an allocation rule - the corresponding
r-optimal mechanisms are uniquely defined.

\subsection{The Main Result}

In this section, we provide our main result, which identifies an r-optimal
allocation rule. To do so, we first propose a general class of ranking
allocation rules. In this generalization, at a generic valuation profile,
the top ranked agent is given the object with some probability $\pi_1$ and
agents ranked $2$ to $\ell$ are given the object with equal probability $%
\pi_2$, where $\pi_1+(\ell-1)\pi_2=1$. Formally, a two-step allocation rule
is defined as follows.

\begin{defn}
A \textbf{two-step ranking} allocation rule is a ranking allocation rule
with probabilities

\begin{equation*}
(\pi_1, \underbrace{\pi_2,\ldots,\pi_2}_{\ell-1},0,\ldots,0),
\end{equation*}
where $\pi_1 > \pi_2 > 0$ and $\pi_1 + (\ell-1)\pi_2=1$.
\end{defn}

Hence, a two-step allocation rule is uniquely defined by $(\pi_1,\ell)$ - $%
\ell$ is the number of agents receiving positive probability. The GL
allocation rule is a two-step ranking allocation rule with $\pi_1=1-1/n$ and 
$\ell=2$. In Proposition \ref{prop:2rank} (see Appendix), we characterize
the class of two-step ranking allocation rules that can be satisfactorily
implemented - this class requires $\ell$ to be even and $\pi_1$ is
determined uniquely given an even value of $\ell$.

We are now ready to state the main result of the paper. It shows that there
is a two-step ranking allocation rule that is r-optimal, which has excellent
convergence to efficiency.

\begin{theorem}
\label{theo:nrgl} There is a two-step ranking allocation rule that is
r-optimal. Its allocation probabilities $(\pi_1^*,\ldots,\pi_n^*)$ are
defined as follows: 
\begin{equation*}
\pi^*_i = \left\{ 
\begin{array}{ll}
1 - \frac{\ell-1}{C(n-2,\ell-1) + \ell } & \text{if $i=1$} \\ 
\frac{1}{C(n-2,\ell-1) + \ell } & \text{if $i \in \{2,\ldots,\ell\}$} \\ 
0 & \text{otherwise,}%
\end{array}
\right.
\end{equation*}
where 
\begin{equation*}
\ell \in \arg \min_{2 \le i \le (n-1),~i~\text{even}} \frac{(i-1)}{\Big(%
C(n-2,i-1) + i \Big)}.
\end{equation*}
Moreover, if $n \ne 8$, there is a unique r-optimal allocation rule.
\end{theorem}

\noindent \textbf{Remark 1.} Though Theorem \ref{theo:nrgl} requires at least three agents, we can easily identify the
r-optimal mechanism in the two-agent case. Proposition %
\ref{prop:rank} continues to hold even if $n=2$. As a result, the only
ranking allocation rule that can be satisfactorily implemented are those
where both the agents get the object with equal probability. Hence, the
unique r-optimal allocation rule is the \emph{equal
sharing} allocation rule where both the agents get the object with
probability $1/2$ - transfers are not needed to make this allocation rule
satisfactorily implementable. \newline

\noindent \textbf{Remark 2.} All our optimality results rely on the fact
that the valuation space $V$ of each agent is \emph{rich} - an interval with
zero as the lowest valuation. We do not know how to extend these results to
a setting where $V$ is an arbitrary interval. However, we stress here that
the mechanism we derive in Theorem \ref{theo:nrgl} remains valid for any
arbitrary interval $V$. To see this, consider $V:=[L,H]$, where $0 \le L < H$%
. Note that our results along with the mechanism in Theorem \ref{theo:nrgl}
hold true if valuation space is $[0,H]$. Now, consider the restriction of
this mechanism to the valuation space $[L,H]$ - such a restriction is
well-defined and satisfactory. Thus, our mechanism will have the same
efficiency properties when $V:=[L,H]$. Of course, this mechanism need not
satisfy the optimality property claimed in Theorem \ref{theo:nrgl} - though,
we have no counter-examples to show this. In fact, we conjecture that
our mechanism will remain optimal even in such type spaces.

\subsection{Computations}

\label{sec:comp}

Besides the optimality of the two-step allocation rule identified in Theorem %
\ref{theo:nrgl}, we want to stress the speed with which it converges to
efficiency. Because of combinatorial terms in the denominator of the
expression for $\pi^*_1$, its convergence to $1$ is exponential. We spell
out the exact nature of this convergence below.

The exact form of the r-optimal allocation rule will depend on the value of $%
n$. Note that the value of $\ell$ is determined by minimizing the following
expression over all even $i \le (n-1)$: 
\begin{equation*}
\min_{2 \le i \le (n-1),~i~\text{even}} \frac{(i-1)}{\Big(C(n-2,i-1) + i %
\Big)}.
\end{equation*}
Routine calculations show that the minimum of this expression occurs when $%
i=2$ for $n < 8$. Hence, for $n < 8$, the GL allocation rule is the unique
r-optimal allocation rule.

If $n=8$, the minimum of this expression occurs at $i=2$ or $i=4$. If $n \ge
9$, the maximum value of $C(n-2,i-1)$ over all even $i$ determines the
minimum of this expression - it is possible that two values of $i$ maximizes 
$C(n-2,i-1)$, in which case we choose the smaller one to minimize $\frac{i-1%
}{C(n-2,i-1)+i}$.

Hence, the choice of $\ell$ in Theorem \ref{theo:nrgl} is unique for all
values of $n \ne 8$. In the proof of Theorem \ref{theo:nrgl}, we show that
as long as we can choose $\ell$ uniquely, the r-optimal allocation rule is
unique.

We now consider the case $n \ge 9$ and give an explicit formula for $\ell$
in this case. Denote by $\lfloor x \rfloor_e$ and $\lfloor x \rfloor_o$
respectively the largest even number smaller than $x$ and the largest odd
number smaller than $x$. We now consider two cases. \newline

\noindent \textsc{Case 1.} If $n$ is odd, then $n-2$ is odd. So, $C(n-2,i-1)$
is maximized at two values of $i-1$: at $\frac{n-2+1}{2}$ or $\frac{n-2-1}{2}
$, out of which one of them is odd. So, we can say $C(n-2,i-1)$ is maximum
when $i-1=\lfloor \frac{n-1}{2} \rfloor_o$ or $i = \lfloor \frac{n+1}{2}
\rfloor_e$. \newline

\noindent \textsc{Case 2.} If $n$ is even, then $C(n-2,i-1)$ is maximum when 
$i-1=\frac{n-2}{2}$. Since we require $(i-1)$ to be odd, we can say that $%
i-1=\lfloor \frac{n-2}{2} \rfloor_o$ or $i=\lfloor \frac{n}{2} \rfloor_e$.
Since $n$ is even, we can equivalently write this as $i=\lfloor \frac{n+1}{2}
\rfloor_e$. \newline

Hence, when $n \ge 9$, we conclude that $\ell$ in Theorem \ref{theo:nrgl} is 
$\lfloor \frac{n+1}{2} \rfloor_e$. We document this as a corollary.

\begin{cor}
\label{cor:comp} The two-step ranking r-optimal allocation rule identified
in Theorem \ref{theo:nrgl} satisfies 
\begin{align*}
\ell = 2~\text{if}~n < 8, \\
\ell \in \{2,4\}~\text{if}~n=8, \\
\ell = \lfloor \frac{n+1}{2} \rfloor_e~\text{if}~n \ge 9.
\end{align*}
Hence, for $n < 8$, the GL allocation rule is the unique r-optimal
allocation rule.
\end{cor}

Corollary \ref{cor:comp} shows that for $n=8$, there are many r-optimal
allocation rules. For $\ell=2$ and $\ell=4$, we have two two-step ranking
allocation rules that are r-optimal. Any convex combination of these two
allocation rules will also be r-optimal. Note that ranking rules generated
by such convex combinations need not be two-step ranking allocation rules.
In conclusion, for $n \ne 8$, we have a unique r-optimal allocation rule
defined by Theorem \ref{theo:nrgl}. But for $n=8$, the uniqueness is lost
and there exists r-optimal allocation rules that are not two-step ranking
allocation rule.

Corollary \ref{cor:comp} allows us to compute the allocation probabilities
of the highest valuation agent using the Pascal triangle in Figure \ref%
{fig:ropt}. Each row (starting with the second row) represents a particular
value of $n$, starting with $n=3$ in the second row. By Corollary \ref%
{cor:comp}, $\ell=2$ if $n < 8$, $\ell \in \{2,4\}$ if $n=8$, and $%
\ell=\lfloor \frac{n+1}{2} \rfloor_e$ if $n > 9$. In each row of the Pascal
triangle, the entries are $C(n-2,0), C(n-2,1), \ldots, C(n-2,n-2)$. Now, the
value $C(n-2,\ell-1)$ is highlighted in the orange (lighter shaded) cell of
each row.~\footnote{%
The values in the brown (darker shaded) cells correspond to the entries of
the Green-Laffont allocation rule.} The probability of the highest valuation
agent is then easily computed from this and the value of $\ell$ as: $\frac{%
C(n-2,\ell-1)+1}{C(n-2,\ell-1)+\ell}$, which is shown to the right of the
Pascal triangle.

Note that for $n \ge 14$, the object is allocated to the highest valuation
agent with at least 99\% probability. The Green-Laffont allocation rule will
require at least 100 agents to achieve such probability for the highest
valuation agent.

\begin{figure}[!hbt]
\centering
\includegraphics[width=6in]{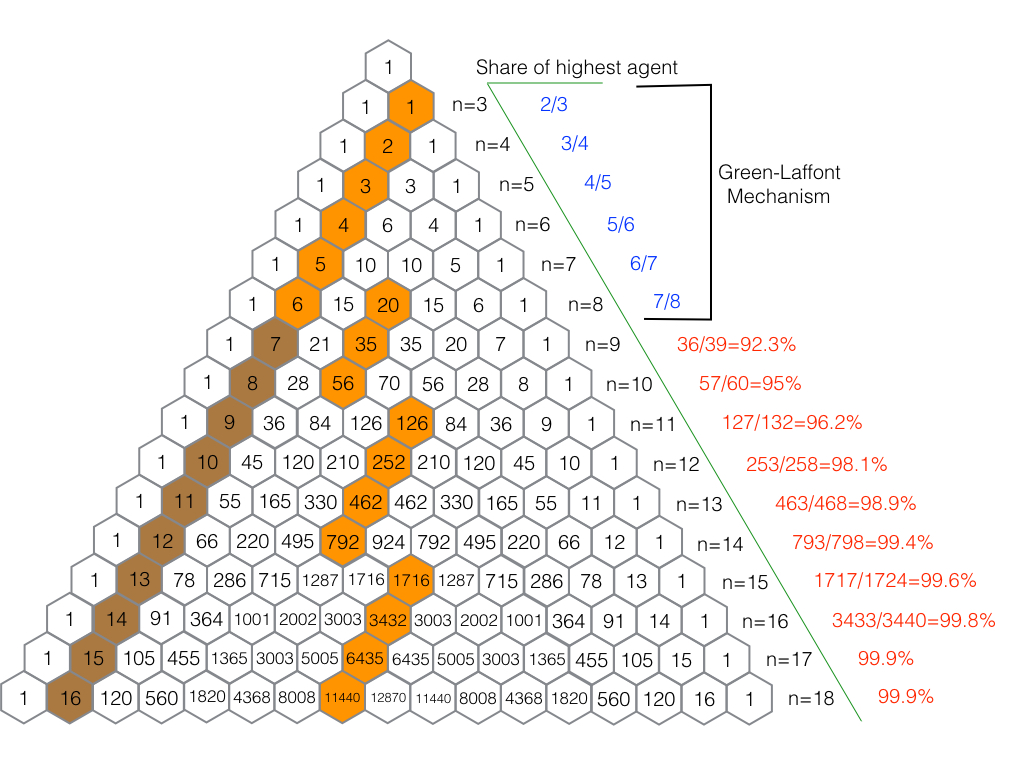}
\caption{The r-optimal allocation rule}
\label{fig:ropt}
\end{figure}

\subsection{Participation Constraints}

We now show that a strong form of participation constraint is satisfied by a
class of ranking mechanisms, including the r-optimal mechanism in Theorem %
\ref{theo:nrgl}.

\begin{defn}
A mechanism $(f,\mathbf{p})$ is \textbf{ex-post individually rational} if
for every $i \in N$ and for every $\mathbf{v}$, we have 
\begin{equation*}
v_i f_i(\mathbf{v}) - p_i(\mathbf{v}) \ge 0.
\end{equation*}
\end{defn}

The ex-post notion of participation constraint is appropriate in our
prior-free model. Notice that, unlike the model in \citet{Cramton87}, our
model does not have any property rights defined for the agents.~\footnote{We discuss
the results in \citet{Cramton87} in details in Section \ref{sec:lit}.} Hence, we
assume that the outside option of each agent is zero. In that sense, even
though our participation constraints are ex-post, they only ensure
non-negative payoff from participation. On the other hand, the participation
constraints in \citet{Cramton87} is interim but because of the property
rights structure, they ensure larger interim payoffs to agents.

We prove below that a \emph{class} of mechanisms using two-step ranking
allocation rules satisfy ex-post individual rationality. For $n \ge 8$, the two extremes of
this class are the Green-Laffont mechanism and our r-optimal mechanism in
Theorem \ref{theo:nrgl}.

\begin{theorem}
\label{theo:ir} Suppose $f$ is a two-step ranking allocation rule defined by 
$(\pi_1,\ell)$, where $2\ell \le n+1$. If $(f,\mathbf{p})$ is a satisfactory
mechanism, then it is ex-post individually rational.
\end{theorem}

The r-optimal allocation rule in Theorem \ref{theo:nrgl} satisfies the
sufficient condition identified in Theorem \ref{theo:ir}.

\begin{cor}
\label{cor:optir} Suppose $f$ is the r-optimal allocation rule identified in
Theorem \ref{theo:nrgl}. If $(f,\mathbf{p})$ is a satisfactory mechanism,
then it is ex-post individually rational.
\end{cor}

\begin{proof}
By Corollary \ref{cor:comp}, the r-optimal allocation rule in Theorem \ref%
{theo:nrgl} satisfies $2\ell \le n+1$. By Theorem \ref{theo:ir}, the claim
follows.
\end{proof}

We compute the payments in the mechanisms discussed in Theorem \ref{theo:ir}%
. While the general payment formula for a satisfactory mechanism is quite
complicated (see Theorem \ref{theo:char}), the payment formula for the
mechanisms in Theorem \ref{theo:ir} is easier to express. \\

\noindent {\bf Notation.} For any pair of positive integers, $K, K'$ with
$K \ge K'$, $$\psi(K',K):= K' \times (K'+1) \times \ldots \times K$$ \\

\begin{prop}
\label{prop:2steppay} Suppose $(f,\mathbf{p})$ is a satisfactory mechanism,
where $f$ is a two-step ranking allocation rule defined by $(\pi_1,\ell)$
with $\pi_1+(\ell-1)\pi_2=1$. For any valuation profile $\mathbf{v}$ with $%
v_1 > v_2 > \ldots > v_n >0$, we have

\begin{itemize}
\item if $i=1$, then 
\begin{equation*}
p_i(\mathbf{v})= -\frac{\pi_2}{(\ell-1)!} \Big[ \sum_{k=1}^{\ell-1}(-1)^k
(k-1)! \psi(n-\ell,n-k-1)v_{k+1} \Big].
\end{equation*}

\item if $i \in \{2,\ldots,\ell\}$, then 
\begin{equation*}
p_i(\mathbf{v})= - \frac{\pi_2}{(\ell-1)!} \Big[ \sum_{k=2}^{i-1}(-1)^k
(k-1)! \psi(n-\ell,n-k-1)v_{k} + \sum_{k=i}^{\ell-1}(-1)^k (k-1)!
\psi(n-\ell,n-k-1) v_{k+1} \Big].
\end{equation*}

\item if $i > \ell$, then 
\begin{equation*}
p_i(\mathbf{v})=-\frac{\pi_2}{(\ell-1)!} \Big[ \sum_{k=2}^{\ell-1}(-1)^k
(k-1)! \psi(n-\ell,n-k-1)v_{k} + (-1)^{\ell} (\ell-1)!v_{\ell}\Big].
\end{equation*}
\end{itemize}
\end{prop}

In any two step ranking allocation rule $(\pi_1,\ell)$, at a valuation
profile $\mathbf{v}$ with $v_1 > v_2 > \ldots > v_n > 0$, an agent $i$ with $%
i > \ell$ gets the object with zero probability - call such agents \emph{%
losing} agents. According to the payment formula computed in Proposition \ref%
{prop:2steppay}, losing agents receive some payments. Theorem \ref{theo:ir}
shows that losing agents receive non-negative payment if $2\ell \le n+1$.
Hence, participation constraints are satisfied for losing agents in such
class of mechanisms. For two step ranking allocation rules, where $2\ell >
n+1$, it is possible that losing agents may be asked to pay, violating their
participation constraint.

\subsection{Pareto Optimal Ranking Mechanisms}

We now discuss an alternate prior-free notion of comparing mechanisms, where
we compare mechanisms at \emph{every} valuation profile in term of total
social welfare. Informally, a satisfactory mechanism $\mathcal{M}$ dominates
another satisfactory mechanism $\mathcal{M}^{\prime }$ if $\mathcal{M}$
generates as much total welfare as $\mathcal{M}^{\prime }$ in every profile
of valuations and strictly higher in some profile of valuations. A
satisfactory mechanism is \emph{Pareto optimal} if it is not dominated by
any other satisfactory mechanism.

It is a relatively weak notion to compare mechanisms - for instance, it may
be that a Pareto optimal mechanism is dominated by another satisfactory
mechanism at a positive measure of valuation profiles. Two satisfactory
mechanisms may not even be comparable using this notion.

We adapt the notion of Pareto optimality to the class of ranking mechanisms.

\begin{defn}
A ranking allocation rule $f$ is \textbf{r-Pareto optimal} if (i) $f$ is
satisfactorily implementable and (ii) there does not exist another ranking
allocation rule $f^{\prime }$ such that $f^{\prime }$ is satisfactorily
implementable and at every valuation profile $\mathbf{v}$, we have 
\begin{equation*}
\sum_{i \in N}v_if^{\prime }_i(\mathbf{v}) \ge \sum_{i \in N}v_if_i(\mathbf{v%
}),
\end{equation*}
with strict inequality holding at some $\mathbf{v}$.

A ranking mechanism $(f,\mathbf{p})$ is r-Pareto optimal if (i) $(f,\mathbf{p%
})$ is a satisfactory mechanism and (ii) $f$ is r-Pareto optimal.
\end{defn}

We first show that the GL allocation rule is an r-Pareto optimal allocation
rule.

\begin{theorem}
\label{theo:rgl} The GL allocation rule is an r-Pareto optimal allocation
rule. Moreover, it is the unique r-Pareto optimal allocation rule satisfying 
$\pi_3=\ldots=\pi_n=0$. \newline
\end{theorem}

Theorem \ref{theo:rgl} gives a foundation for the GL mechanism. Among all
ranking mechanisms that only allocate the object to top-two agents, the GL
mechanism is the unique r-Pareto optimal mechanism. As we show in the next
result, if $n \le 8$, the GL mechanism is the unique r-Pareto optimal
mechanism, but there are other r-Pareto optimal mechanisms if the number of
agents is greater than $8$. In particular, our r-optimal mechanism is always
r-Pareto optimal.

\begin{prop}
\label{prop:rpar} For $n \le 8$, the GL allocation rule is the unique
r-Pareto optimal allocation rule. For $n > 8$, the unique r-optimal
allocation rule identified in Theorem \ref{theo:nrgl} is also r-Pareto
optimal. Further, for any arbitrary r-Pareto optimal allocation rule $%
(\pi_1,\ldots,\pi_n)$, we have 
\begin{equation*}
1-1/n \le \pi_1 \le \pi_1^*,
\end{equation*}
where $\pi_1^*$ is as defined in Theorem \ref{theo:nrgl}.
\end{prop}

\section{Satisfactory Implementability}

\label{sec:proof}

In this section, we provide a characterization that drives all our main
results. In particular, we provide a complete characterization of allocation
rules which can be satisfactorily implemented. Besides the technical aspect,
there are other reasons why such a characterization is useful: (1) it
provides a recipe for carrying out such analysis of satisfactory mechanisms
in other models and (2) it showcases the rich but complex class of
non-ranking mechanisms that are satisfactory, thus, highlighting the
salience of ranking mechanisms.

Before stating the characterization, we remind the reader about the
following characterization of DSIC mechanisms by Myerson.~\footnote{%
The characterization in Myerson is for Bayesian incentive compatible
mechanisms. However, it is straightforward to extend it to DSIC mechanisms.}

\begin{lemma}[\citet{Myerson81}]
\label{lem:my} A mechanism $(f,\mathbf{p})$ is DSIC if and only if

\begin{itemize}
\item \textbf{Montonicity of} $f$. for every $i \in N$, for every $v_{-i}
\in V^{n-1}$, and for every $v_i,v^{\prime }_i \in V$ with $v_i > v^{\prime
}_i$, we have 
\begin{equation*}
f_i(v_i,v_{-i}) \ge f_i(v^{\prime }_i,v_{-i}).
\end{equation*}

\item \textbf{Revenue Equivalence.} for every $i \in N$, for every $v_{-i}
\in V^{n-1}$, and for every $v_i \in V$, we have 
\begin{equation*}
p_i(v_i,v_{-i}) = p_i(0,v_{-i}) + v_if_i(v_i,v_{-i}) -
\int_0^{v_i}f_i(x_i,v_{-i})dx_i.
\end{equation*}
\end{itemize}
\end{lemma}

For any mechanism $M \equiv (f,\mathbf{p})$, we define $\mathcal{U}^M_i(%
\mathbf{v})$ as the net utility of agent $i$ at valuation profile $\mathbf{v}
$: 
\begin{equation*}
\mathcal{U}^M_i(\mathbf{v})= v_if_i(\mathbf{v}) - p_i(\mathbf{v}).
\end{equation*}
A consequence of the Myersonian characterization of DSIC is the following
characterization of DSIC and budget-balanced mechanisms.

\begin{prop}
\label{prop:myerson} A mechanism $M \equiv (f,\mathbf{p})$ is DSIC and
budget-balanced if and only if

\begin{enumerate}
\item for every $i \in N$, for every $v_{-i} \in V^{n-1}$, and for every $%
v_i, v^{\prime }_i \in V$ with $v_i > v^{\prime }_i$ we have 
\begin{align*}
f_i(v_i,v_{-i}) \ge f_i(v^{\prime }_i,v_{-i}).
\end{align*}

\item for every $i \in N$, for every $v_{-i} \in V^{n-1}$, for every $v_i
\in V$, we have 
\begin{align*}
\mathcal{U}_i^M(v_i,v_{-i}) &=\mathcal{U}_i^M(0,v_{-i}) + \int_0^{v_i}
f_i(x_i,v_{-i})dx_i.
\end{align*}

\item for every $\mathbf{v} \equiv (v_1,\ldots,v_n) \in V^n$, 
\begin{align*}
\sum_{i \in N}\mathcal{U}_i^M(0,v_{-i}) &= \sum_{i \in N}\big[v_i f_i(%
\mathbf{v}) - \int_0^{v_i} f_i(x_i,v_{-i})dx_i \big].
\end{align*}
\end{enumerate}
\end{prop}

\begin{proof}
From Lemma \ref{lem:my}, (1) and (2) are equivalent to DSIC. For (3), note
that budget-balance of a mechanism $M \equiv (f,\mathbf{p})$ requires that
for all $\mathbf{v} \equiv (v_1,\ldots,v_n) \in V^n$, we must have 
\begin{align*}
\sum_{i \in N}\mathcal{U}^M_i(\mathbf{v}) &= \sum_{i \in N}v_i f_i(\mathbf{v}%
).
\end{align*}
Using (2), we conclude that a DSIC mechanism is budget-balanced if and only
if for all $\mathbf{v} \equiv (v_1,\ldots,v_n) \in V^n$, 
\begin{align*}
\sum_{i \in N}\mathcal{U}_i^M(0,v_{-i}) + \sum_{i \in N}\int_0^{v_i}
f_i(x_i,v_{-i})dx_i &= \sum_{i \in N}v_i f_i(\mathbf{v}).
\end{align*}
Equivalently, a DSIC mechanism $M \equiv (f,\mathbf{p})$ is budget-balanced
if and only if for all $\mathbf{v} \equiv (v_1,\ldots,v_n) \in V^n$, 
\begin{align*}
\sum_{i \in N}\mathcal{U}_i^M(0,v_{-i}) &= \sum_{i \in N}\big[v_i f_i(%
\mathbf{v}) - \int_0^{v_i} f_i(x_i,v_{-i})dx_i \big].
\end{align*}
\end{proof}

Our main characterization, like Myerson's characterization, provides a way
to separate out the allocation rule and the payment rule in a satisfactory
mechanism. While Myerson does not impose budget-balance, our result shows
that this separation continues to hold even if we impose budget-balance.

Fix an allocation rule $f$. If $f$ is monotone (in the sense of Lemma \ref%
{lem:my}), then we can immediately define a payment scheme $\mathbf{p}$ that
makes $(f,\mathbf{p})$ DSIC as follows: for every $i \in N$ and for every $%
\mathbf{v}$, set 
\begin{equation*}
p_i(\mathbf{v})= v_i f_i(\mathbf{v}) - \int_0^{v_i}f_i(x_i,v_{-i})dx_i.
\end{equation*}
Note that $p_i(0,v_{-i})=0$ for all $i$ and for all $v_{-i}$ in this
mechanism. We call a mechanism defined from such a payment scheme as the 
\textbf{elementary mechanism} corresponding to a monotone $f$. It can be
easily verified that if $f$ is the efficient allocation rule, then the
corresponding elementary mechanism is the Vickrey auction.

For every valuation profile $\mathbf{v}$, define for every $i \in N$, the
payment of agent $i$ in the elementary mechanism corresponding to a monotone 
$f$ as: 
\begin{equation*}
R^f_i(\mathbf{v}):= v_i f_i(\mathbf{v}) - \int_0^{v_i}f_i(x_i,v_{-i})dx_i.
\end{equation*}
Then, 
\begin{equation*}
R^f(\mathbf{v}):=\sum_{i \in N}R^f_i(\mathbf{v}),
\end{equation*}
denotes the total revenue collected at valuation profile $\mathbf{v}$ in the
elementary mechanism corresponding to $f$.

We will provide necessary and sufficient conditions on $f$ for it to be
satisfactorily implementable. These conditions are given in terms of revenue
collected from the elementary mechanism corresponding to $f$ at various
valuation profiles.

At any valuation profile $\mathbf{v}$, define $N^0_{\mathbf{v}}:=\{i \in N:
v_i=0\}$. Given any valuation profile $\mathbf{v}$, for any $T \subseteq N$,
we denote by $(0_T,v_{-T})$ the valuation profile where all the agents in $T$
have zero valuation and each agent $i \notin T$ has valuation $v_i$.

\begin{defn}
An allocation rule $f$ is \textbf{residually balanced} if for every $\mathbf{%
v}$ such that $N^0_{\mathbf{v}} = \emptyset$, we have 
\begin{align}  \label{eq:balance}
\sum_{T \subseteq N}(-1)^{|T|}~R^f(0_T,v_{-T}) &= 0.
\end{align}
\end{defn}

Residual balancedness is a technical combinatorial condition on an
allocation rule. We show that for a symmetric and monotone allocation rule
residual balancedness is necessary and sufficient for satisfactory
implementability.

\begin{theorem}
\label{theo:char} A symmetric allocation rule $f$ is satisfactorily
implementable if and only if it is (a) monotone and (b) residually balanced. 
\newline

Further, if $f$ is satisfactorily implementable, then there is a unique $%
\mathbf{p}$ such that $(f,\mathbf{p})$ is a satisfactory mechanism. Such a
unique $\mathbf{p}$ is defined as follows: for all $\mathbf{v} \in V^n$, for
all $i \in N$, 
\begin{align*}
p_i(\mathbf{v}) &= - \frac{1}{|N^0_{\mathbf{v}}|} \sum_{T \subseteq N: N^0_{%
\mathbf{v}} \subseteq T} \frac{(-1)^{|T \setminus N^0_{\mathbf{v}} |}}{%
C(|T|,|N^0_{\mathbf{v}}|)} R^f(0_T,v_{-T}) ~\qquad~\text{if $i \in N^0_{%
\mathbf{v}}$} \\
p_i(\mathbf{v}) &= R^f_i(\mathbf{v}) - \frac{1}{|N^0_{\mathbf{v}}|+1}
\sum_{T \subseteq N: (N^0_{\mathbf{v}} \cup \{i\}) \subseteq T} \frac{%
(-1)^{|T \setminus N^0_{\mathbf{v}} | - 1}}{C(|T|,(|N^0_{\mathbf{v}}|+1))}%
R^f(0_T,v_{-T})~\qquad~\text{if $i \notin N^0_{\mathbf{v}}$}
\end{align*}
\end{theorem}

The condition in Theorem \ref{theo:char} looks very similar to the cubical
array lemma in \citet{Walker80}. While the cubical array lemma applies to
only efficient allocation rule, our characterization is for \emph{any}
allocation rule. Theorem 2 in \citet{Yenmez15} characterizes ex-post
incentive compatible and budget-balanced mechanisms.~\footnote{%
His solution concept is ex-post incentive compatibility because he looks at
a setting that can potentially allow for interdependent valuations.} His
characterization is a characterization of DSIC and budget-balanced \emph{%
mechanisms}, and hence, still uses transfers. On the other hand, the
advantage of our characterization is that it gives necessary and sufficient
condition on the \emph{allocation rule} to be satisfactorily implementable.
Thus, we are able to separate out allocation rule and payments for analyzing
budget-balanced mechanisms.

The proof of Theorem \ref{theo:char} is in the Appendix. It is notationally
quite complex. Here, we illustrate the idea of the necessity part with an
example of three agents: $N=\{1,2,3\}$. Let $f$ be a symmetric, monotone,
and satisfactorily implementable allocation rule. Then, there is a $\mathbf{p%
}$ such that $(f,\mathbf{p})$ is a satisfactory mechanism. Consider a
valuation profile $\mathbf{v} \equiv (0,0,0)$. By BB and symmetry, we get $%
p_1(\mathbf{v})=p_2(\mathbf{v})=p_3(\mathbf{v})=0$. Now, consider a
valuation profile $\mathbf{v} \equiv (v_1,0,0)$. By Lemma \ref{lem:my}, 
\begin{equation*}
p_1(\mathbf{v})=p_1(0,0,0)+ R_1^f(\mathbf{v})=R_1^f(\mathbf{v}).
\end{equation*}
Note that $R^f(\mathbf{v})=R^f_1(\mathbf{v})$. By symmetry $p_2(\mathbf{v}%
)=p_3(\mathbf{v})$. Hence, by BB and symmetry, 
\begin{equation*}
0=p_1(\mathbf{v})+2p_2(\mathbf{v}) = 2p_2(\mathbf{v}) + R^f(\mathbf{v}).
\end{equation*}
This implies that 
\begin{equation*}
p_2(v_1,0,0)=-\frac{1}{2}R^f(v_1,0,0).
\end{equation*}
Now, consider a valuation profile $\mathbf{v} \equiv (v_1,v_2,0)$. Using BB
and Lemma \ref{lem:my}, and following the above arguments, we get 
\begin{align*}
p_1(\mathbf{v})= p_1(0,v_2,0) + R^f_1(\mathbf{v})=-\frac{1}{2} R^f(0,v_2,0)
+ R^f_1(\mathbf{v}) \\
p_2(\mathbf{v})= p_2(v_1,0,0) + R^f_2(\mathbf{v})=-\frac{1}{2} R^f(v_1,0,0)
+ R^f_2(\mathbf{v}) \\
\end{align*}
Adding these two with $p_3(\mathbf{v})$ and using BB, we get 
\begin{equation*}
p_3(v_1,v_2,0) = \frac{1}{2} \Big( R^f(v_1,0,0) + R^f(0,v_2,0) \Big) -
R^f(v_1,v_2,0).
\end{equation*}
Finally, consider the valuation profile $(v_1,v_2,v_3)$ with $v_1,v_2,v_3 >
0 $. Again, using Lemma \ref{lem:my}, we get 
\begin{align*}
p_1(\mathbf{v})= p_1(0,v_2,v_3) + R^f_1(\mathbf{v})=\frac{1}{2} \Big( %
R^f(0,v_2,0) + R^f(0,0,v_3) \Big) - R^f(0,v_2,v_3) + R^f_1(\mathbf{v}) \\
p_2(\mathbf{v})= p_2(v_1,0,v_3) + R^f_2(\mathbf{v})=\frac{1}{2} \Big( %
R^f(v_1,0,0) + R^f(0,0,v_3) \Big) - R^f(v_1,0,v_3) + R^f_2(\mathbf{v}) \\
p_3(\mathbf{v})= p_3(v_1,v_2,0) + R^f_2(\mathbf{v})=\frac{1}{2} \Big( %
R^f(0,v_2,0) + R^f(v_1,0,0) \Big) - R^f(v_1,v_2,0) + R^f_3(\mathbf{v}) \\
\end{align*}
Adding and using BB, we get 
\begin{equation*}
R^f(v_1,v_2,v_3) - R^f(v_1,v_2,0) - R^f(0,v_2,v_3) - R^f(v_1,0,v_3) +
R^f(v_1,0,0) + R^f(0,v_2,0) + R^f(0,0,v_3) = 0,
\end{equation*}
which is the residual balancedness condition. The sufficiency can be shown
using the explicit form of payment functions hidden in these calculations.
In summary, residual balancedness allows a recursive calculation of payments
at all valuation profiles so that budget-balance holds.

\section{Relation to the Literature}

\label{sec:lit}

The impossibility of achieving efficiency, dominant strategy incentive
compatibility, and budget-balance was first shown by \citet{Green79}, which
also contains a lot of discussions on achieving \emph{second-best} using
non-efficient but DSIC and budget-balanced mechanisms. This includes the
Green-Laffont mechanism that we discuss. Though, they focussed attention on
public good problems and gave sketches of the Green-Laffont mechanism we
discuss, they clearly anticipated the mechanism as well as many
environments beyond the public good problem where the impossibility result
would hold. \citet{Gary00} contains an extensive discussion on this - they
also formally define the Green-Laffont mechanism and study its statistical
and strategic properties in the public good problem.

This impossibility result started a long literature on how to overcome it.
We classify them in several categories and discuss some relevant ones. Most
of the literature we discuss concern with private good allocation among
several buyers. There are parallel literature on bilateral trading and
public good provision that we do not discuss. \newline

\noindent \textsc{Domain identification.} Classic revenue equivalence
results imply that every efficient and DSIC mechanism must be a Groves
mechanism~\citep{Green77,Holmstrom79}. The Green-Laffont impossibility result
essentially implies that no Groves mechanism can balance budget in many
settings - though their focus is mainly of public good problems. In the
public good context, \citet{Laffont80} consider differentiable mechanisms
and show that existence of a DSIC, BB, and efficient mechanism is equivalent
to solving a system of differential equations. In the same model, %
\citet{Walker80} identifies domains (of utility functions of agents) where
impossibilities exist - he restricts attention to continuous mechanisms.
As corollary of their results, they identify form of
utility functions of agents where possibility or impossibility result exists.
\citet{Hurwicz90} extend the Green-Laffont impossibility to pure exchange economies.
These papers are mainly focused on identifying domains where the negative result
of Green and Laffont persists.

But there are settings where DSIC, BB, and efficient mechanisms exist. %
\citet{Suijs96} is a good example of a domain where Groves mechanisms that
balance the budget exists - he discusses a \emph{sequencing problem}. In the
context of multi-object assignment, a recent contribution is \citet{Mitra10}%
. This paper identifies domains of multi-object auctions where the
Green-Laffont impossibility can be overcome. \newline

\noindent \textsc{Bayesian incentive compatibility.} One way to get around
the Green-Laffont impossibility is to consider the weaker solution concept
of Bayesian incentive compatibility. \citet{Arrow79,dAspremont79} construct
Bayesian incentive compatible, efficient, and budget-balanced mechanism, now
known as the dAGV mechanism, that work in a variety of settings. The dAGV
mechanisms fail to be interim individually rational in many settings. In an
unpublished paper, \citet{Fudenberg95} extend this result in the following
sense - for every Bayes-Nash implementable allocation rule, there exists a
Bayesian incentive compatible and budget-balanced mechanism using this
allocation rule. Like in the dAGV mechanism, such budget-balanced mechanisms
need not satisfy interim individual rationality. \citet{Rahman11} gives a
characterization of Bayesian (and ex-post) incentive compatible and
budget-balanced mechanisms in a very general framework.

In a seminal paper, \citet{Cramton87} show that efficient, Bayesian
incentive compatible, budget-balanced mechanisms satisfying interim
individual rationality is possible in a single object allocation problem.~%
\footnote{%
They consider a problem where agents have property rights over the object,
and stronger form of interim individual rationality is satisfied by their
mechanism. However, their results can still be applied to our problem if we
assume equal property rights to all the agents.} The possibility result in
our problem using Bayesian incentive compatibility is in sharp contrast to
the impossibility results known in bilateral trading problems like in ~%
\citet{Myerson83}.

Unlike \citet{Cramton87}, we focus on DSIC mechanisms, and our mechanism is
not efficient. Naturally, the mechanism in \citet{Cramton87} require a lot
of prior information. Our mechanism is prior-free and satisfies
ex-post individually rationality. Thus, we illustrate a prior-free way of approximately achieving the
possibility result in \citet{Cramton87}. \newline

\noindent \textsc{Redistribution mechanisms.} The prior-free approach of
mechanism design using DSIC mechanisms have been popular in algorithmic game
theory literature in computer science. Restricting attention to efficient
mechanisms, which means restricting attention to Groves mechanisms, several
papers relax budget-balance and show how best to redistribute the surplus
revenue. The measure of efficiency of redistribution is \emph{worst-case} in
these papers. One of the earliest papers to do this is \citet{Cavallo06},
who studied this problem in our setting (single object allocation). He
showed that remarkable Groves mechanisms exist that can redistribute large
fraction of Vickrey auction payments using Groves mechanisms. %
\citet{Moulin09} and \citet{Guo09} derive optimal redistribution mechanisms
in the multi-unit allocation setting where agents demand exactly one unit -
their mechanisms are identical and discovered independently.~\footnote{%
Several papers related to this theme have also appeared - see for instance, %
\citet{Apt08} and \citet{Moulin10}.} As the number of agents increase, like
our mechanism, their Groves mechanisms can redistribute large fraction of
Vickrey auction revenue among agents. The main difference from these papers
and ours is budget-balance. Since these papers do not impose budget-balance,
the actual budget imbalance in these mechanisms can be high in various
valuation profiles. On the other hand, like in \citet{Cramton87},
budget-balance is a constraint in our problem. Hence, unlike these papers,
we work with mechanisms outside the Groves class. Our results show that we
can achieve excellent levels of efficiency (99\% with at least 14 agents)
using DSIC and budget-balanced mechanisms. \newline

\noindent \textsc{Beyond Groves mechanisms.} While most of the literature
seems to have either weakened DSIC to Bayesian incentive compatibility or
relaxed the budget-balanced criteria while working with efficient and DSIC
mechanisms (Groves mechanisms), there is very little literature on exploring
the limits of DSIC and budget-balanced mechanisms. We do this for the case
of single object allocation problem. One of the problems with exploring
non-Groves mechanisms is that we search over the space of allocation rules
and payment rules - Groves mechanisms pin down the allocation rule to be the
efficient allocation rule. A non-efficient allocation rule can achieve
better social welfare redistribution is well known - see for instance
examples in \citet{Laffont80} and a more computational analysis in %
\citet{deClippel09}. \citet{Sprumont13} consider Pareto-undominated
mechanisms by considering DSIC and non-efficient mechanisms, though his
mechanisms are not budget-balanced. \citet{Faltings05} and \citet{Guo11}
consider variants of Green-Laffont mechanisms discussed in \citet{Green79}
and show some worst-case results, but they do not consider the general class
of DSIC and budget-balanced mechanisms that we analyze. \citet{Hashimoto15}
discusses a non-ranking satisfactory mechanism and provides several
axiomatization his mechanism.

Another possibility is to consider priors and design the expected welfare
maximizing DSIC and budget-balanced mechanism for allocating an object. This
is similar to the expected revenue maximizing question in \citet{Myerson81},
but significantly more complicated. Restricting attention to the case of two
agents and deterministic mechanisms, \citet{Drexl15} derive the optimal
expected welfare maximizing DSIC and budget-balanced mechanism. %
\citet{Shao13} do the same analysis for two agents but without requiring
budget-balance. These papers illustrate difficulty in extending such
analysis to more than two agents. In that sense, we provide a prior-free
method of measuring welfare of mechanisms which turns out to be tractable
for any number of agents.

\section{Conclusion}

\label{sec:conc}

This paper provides a novel DSIC, budget-balanced, symmetric, and ex-post
individually rational mechanism to allocate a single unit of a resource. The
mechanism converges to efficiency with moderately high number of agents.
Further, the mechanism can be viewed as a generalization of the
Green-Laffont mechanism. From a methodological standpoint, we provide
several key insights on how to analyze DSIC and budget-balanced mechanisms,
and propose a tractable class of mechanisms that we call ranking mechanisms.

While we carry out this analysis for allocating a single unit of resource,
we feel that the ideas in this paper can be pushed in other models of
mechanism design where budget-balance is a constraint. Further, an indirect
implementation of our mechanism will significantly improve the practicality
of our proposed mechanism.

From a broader perspective, our results quantify the impossibility on
designing DSIC, budget-balanced, and efficient mechanisms in the single
object allocation problem. It shows that even though impossibility exists,
it is really thin. Thus, the possibility results with Bayesian incentive
compatibility \citep{Cramton87} or approximate possibility results with
relaxed budget-balanced constraints \citep{Guo09,Moulin09} can also be
approximately achieved with DSIC and budget-balanced mechanisms.

\newpage

\section*{Appendix: Omitted Proofs}

This section contains all the missing proofs. We first prove our workhorse
result - Theorem \ref{theo:char}. Once this result is proved, we use it to
prove Proposition \ref{prop:rank}. Then, we proceed to prove our two main
results - Theorem \ref{theo:rgl} and Theorem \ref{theo:nrgl}. Then, we prove
our individual rationality result - Theorem \ref{theo:ir}. \newline

\noindent \textbf{Notations.} We will need some extra notations. At every
valuation profile $\mathbf{v}$ and for every $k \in N$, we denote by $%
v_{(k)} $ the valuation of every agent in $\mathbf{v}[k]$. Note that for
some $k \in N$, it is possible that $\mathbf{v}[k] = \emptyset$, in which
case $v_{(k)}$ is defined to be $0$. For any $j \in N$, let the cardinality
of the set $\cup_{h=1}^{j}\mathbf{v}[h]$ be $L_j$.

We illustrate these notations with an example. Suppose $N=\{1,\ldots,8\}$.
Consider a valuation profile $\mathbf{v}$ such that $\mathbf{v}[1]=\{1,2\}$, 
$\mathbf{v}[2]=\{3,4,5,6\}$, and $\mathbf{v}[3]=\{7,8\}$. Then, $L_1=2,
L_2=6, L_3=8$. According to a ranking allocation rule with probabilities $%
(\pi_1,\ldots,\pi_8)$, agents $1$ and $2$ share $\pi_1+\pi_2$ equally,
agents $3,4,5,6$ share $\pi_3+\pi_4+\pi_5+\pi_6$ equally and agents $7$ and $%
8$ share $\pi_7+\pi_8$ equally. In other words, for every $j \in N$, agents
in $\mathbf{v}[j]$ share equally the probabilities 
\begin{equation*}
\pi_{L_{j-1}+1} + \ldots + \pi_{L_j},
\end{equation*}
where $L_0 \equiv 0$.

We begin by a lemma, which will be useful to all our proofs.

\begin{lemma}
\label{lem:contin} Suppose $f$ is a ranking allocation rule. Then, $R^f$ is
continuous.
\end{lemma}

\begin{proof}
For any $\mathbf{v}$, we know that 
\begin{equation*}
R^f(\mathbf{v})= \sum_{i \in N}v_if_i(\mathbf{v}) - \sum_{i \in
N}\int_0^{v_i}f_i(x_i,v_{-i})dx_i.
\end{equation*}

We now do the proof in two steps. Assume that the allocation probabilities
of the ranking allocation rule are $\pi_1 \ge \pi_2 \ge \ldots \ge \pi_n$. 
\newline

\noindent \textsc{Step 1.} In this step, we show that for every $i\in N$,
the expression $\int_{0}^{v_{i}}f_{i}(x_{i},v_{-i})dx_{i}$ is continuous in $%
\mathbf{v}$. Fix a valuation profile $\mathbf{v}$. Consider agent $i$ and
suppose $i\in \mathbf{v}[j]$. Hence, $v_{i}\equiv v_{(j)}$ for some $j$.
Then, using the definition of the ranking allocation rule, we note that 
\begin{align*}
\int_{0}^{v_{i}}f_{i}(x_{i},v_{-i})dx_{i}&
=\int_{0}^{v_{(j)}}f_{i}(x_{i},v_{-i})dx_{i} \\
& =\pi _{L_{j}}(v_{(j)}-v_{(j+1)})+\pi
_{L_{j+1}}(v_{(j+1)}-v_{(j+2)})+\ldots  \\
& =\sum_{h\geq j}\pi _{L_{h}}(v_{(h)}-v_{(h+1)}).
\end{align*}%
To establish continuity, we look at a valuation profile $\mathbf{v}^{\prime }
$ which is arbitrarily close to $\mathbf{v}$, and $\mathbf{v}^{\prime }$ and 
$\mathbf{v}$ differ in valuation of only agent $k$ - it is enough to look at
valuation profiles that differ in one component. Suppose $k\in \mathbf{v}%
[\ell ]$. If $\ell <j$, then there is nothing to prove since the above
summation is unchanged from $\mathbf{v}$ to $\mathbf{v}^{\prime }$. Hence,
assume $\ell \geq j$. Since $\mathbf{v}^{\prime }$ is arbitrarily close to $%
\mathbf{v}$, it must be that $k\in \mathbf{v}^{\prime }[\ell ]$ (this
happens if $v_{k}^{\prime }>v_{k}$) or $k\in \mathbf{v}^{\prime }[\ell +1]$
(this happens if $v_{k}^{\prime }<v_{k}$). Indeed, since $\mathbf{v}^{\prime
}$ is arbitrarily close to $\mathbf{v}$, it must be that $\{k\}=\mathbf{v}%
^{\prime }[\ell +1]$ or $\{k\}=\mathbf{v}^{\prime }[\ell ]$. We consider
both the cases separately. We denote the cardinality of the set $\cup
_{h=1}^{r}\mathbf{v}^{\prime }[h]$ by $L_{r}^{\prime }$ for all $r$. Note
that if $\{k\}=\mathbf{v}[\ell ]$ (i.e., if $k$ is the only element in $%
\mathbf{v}[\ell ]$), then $L_{r}=L_{r}^{\prime }$ for all $r$. As a result, 
\begin{align*}
\int_{0}^{v_{i}^{\prime }}f_{i}(x_{i},v_{-i}^{\prime })dx_{i}& =\sum_{h\geq
j}\pi _{L_{h}^{\prime }}(v_{(h)}^{\prime }-v_{(h+1)}^{\prime }) \\
& \rightarrow \sum_{h\geq j}\pi _{L_{h}}(v_{(h)}-v_{(h+1)}) \\
& =\int_{0}^{v_{i}}f_{i}(x_{i},v_{-i})dx_{i}
\end{align*}%
So, the claim is true. Hence, we assume $|\mathbf{v}[\ell ]|>1$. Now,
consider the following two cases. \newline

\noindent \textsc{Case 1-a.} Suppose $\{k\} = \mathbf{v}^{\prime }[\ell]$.
Since $|\mathbf{v}[\ell]| > 1$, $v^{\prime }_{(\ell)}=v^{\prime }_k
\rightarrow v_k=v_{(\ell)}=v^{\prime }_{(\ell+1)}$. Then, $L^{\prime }_r=L_r$
for all $r < \ell$ and $L^{\prime }_r=L_{r-1}$ for all $r > \ell$. Further, $%
v^{\prime }_{(r)}=v_{(r)}$ for all $r < \ell$ and $v^{\prime }_{(r)}
=v_{(r-1)}$ for all $r > \ell$. As a result, 
\begin{align*}
\int_0^{v^{\prime }_i}f_i(x_i,v^{\prime }_{-i})dx_i &= \sum_{h \ge j}
\pi_{L^{\prime }_h} (v^{\prime }_{(h)} - v^{\prime }_{(h+1)}) \\
&= \sum_{h=j}^{\ell-1}\pi_{L^{\prime }_h}(v^{\prime }_{(h)} - v^{\prime
}_{(h+1)}) + \pi_{L^{\prime }_{\ell}}(v^{\prime }_{(\ell)}-v^{\prime
}_{(\ell+1)}) + \sum_{h \ge \ell+1} \pi_{L^{\prime }_h}(v^{\prime }_{(h)} -
v^{\prime }_{(h+1)}) \\
&\rightarrow \sum_{h=j}^{\ell-1}\pi_{L_h}(v_{(h)} - v_{(h+1)}) + \sum_{h \ge
\ell+1} \pi_{L_{h-1}}(v_{(h-1)} - v_{(h)}) \\
&= \sum_{h=j}^{\ell-1}\pi_{L_h}(v_{(h)} - v_{(h+1)}) + \sum_{h \ge \ell}
\pi_{L_{h}}(v_{(h)} - v_{(h+1)}) \\
&= \int_0^{v_i}f_i(x_i,v_{-i})dx_i.
\end{align*}

This shows that $\int_0^{v_i}f_i(x_i,v^{\prime }_{-i})dx_i \rightarrow
\int_0^{v_i}f_i(x_i,v_{-i})dx_i$ as $v^{\prime }_k \rightarrow v_k$. \newline

\noindent \textsc{Case 1-b.} Suppose $\{k\} = \mathbf{v}^{\prime }[\ell+1]$.
Since $|\mathbf{v}[\ell]| > 1$, we have $v^{\prime }_{(\ell)}=v_{(\ell)}$.
This implies that $v^{\prime }_{(\ell+1)}=v^{\prime }_k \rightarrow
v_k=v_{(\ell)}=v^{\prime }_{(\ell)}$.

Here, we need to worry about the case $k=i$. If $k=i$, then $i \in \mathbf{v}%
^{\prime }[\ell+1]$. Further, for every $r \ge \ell+1$, we have $L^{\prime
}_{r}=L_{r-1}$ and for every $r > \ell+1$, we have $v^{\prime
}_{(r)}=v_{(r-1)}$. 
\begin{align*}
\int_0^{v^{\prime }_i}f_i(x_i,v^{\prime }_{-i})dx_i &= \sum_{h \ge \ell+1}
\pi_{L^{\prime }_h} (v^{\prime }_{(h)} - v^{\prime }_{(h+1)}) \\
&= \pi_{L^{\prime }_{\ell+1}} (v^{\prime }_{(\ell+1)} - v^{\prime
}_{(\ell+2)}) + \sum_{h > \ell+1} \pi_{L^{\prime }_h} (v^{\prime }_{(h)} -
v^{\prime }_{(h+1)}) \\
&\rightarrow \pi_{L_{\ell}}(v_{(\ell)}-v_{(\ell+1)}) + \sum_{h > \ell+1}
\pi_{L_{h-1}} (v_{(h-1)} - v_{(h)}) \\
&= \sum_{h \ge \ell+1} \pi_{L_{h-1}} (v_{(h-1)} - v_{(h)}) \\
&= \sum_{h \ge \ell} \pi_{L_{h}} (v_{(h)} - v_{(h+1)}) \\
&= \int_0^{v_i}f_i(x_i,v_{-i})dx_i.
\end{align*}
This shows that $\int_0^{v_i}f_i(x_i,v^{\prime }_{-i})dx_i \rightarrow
\int_0^{v_i}f_i(x_i,v_{-i})dx_i$ as $v^{\prime }_i \rightarrow v_i$.

A similar proof works if $k \ne i$. Then, $L^{\prime }_r=L_r$ for all $j <
\ell$ and $L^{\prime }_{r}=L_{r-1}$ for all $r > \ell$. Further, $v^{\prime
}_{(r)}=v_{(r)}$ for all $r < \ell$ and $v^{\prime }_{(r)}=v_{(r-1)}$ for
all $r > \ell$. As a result, 
\begin{align*}
\int_0^{v^{\prime }_i}f_i(x_i,v^{\prime }_{-i})dx_i &= \sum_{h \ge j}
\pi_{L^{\prime }_h} (v^{\prime }_{(h)} - v^{\prime }_{(h+1)}) \\
&= \sum_{h=j}^{\ell-1}\pi_{L^{\prime }_h}(v^{\prime }_{(h)} - v^{\prime
}_{(h+1)}) + \pi_{L^{\prime }_{\ell}}(v^{\prime }_{(\ell)}-v^{\prime
}_{(\ell+1)}) + \sum_{h \ge \ell+1} \pi_{L^{\prime }_h}(v^{\prime }_{(h)} -
v^{\prime }_{(h+1)}) \\
&\rightarrow \sum_{h=j}^{\ell-1}\pi_{L_h}(v_{(h)} - v_{(h+1)}) + \sum_{h \ge
\ell+1} \pi_{L_{h-1}}(v_{(h-1)} - v_{(h)}) \\
&= \sum_{h=j}^{\ell-1}\pi_{L_h}(v_{(h)} - v_{(h+1)}) + \sum_{h \ge \ell}
\pi_{L_{h}}(v_{(h)} - v_{(h+1)}) \\
&= \int_0^{v_i}f_i(x_i,v_{-i})dx_i.
\end{align*}

This again shows that $\int_0^{v_i}f_i(x_i,v^{\prime }_{-i})dx_i \rightarrow
\int_0^{v_i}f_i(x_i,v_{-i})dx_i$ as $v^{\prime }_k \rightarrow v_k$. \newline

\noindent \textsc{Step 2.} Now, we argue that the summation $\sum_{i \in
N}v_if_i(\mathbf{v})$ is continuous. Fix a valuation profile $\mathbf{v}$.
Consider all the $j$-th ranked valuation agents, $\mathbf{v}[j]$, for some $%
j $. Note that the total sum of welfare of agents in $\mathbf{v}[j]$ is 
\begin{equation*}
v_{(j)} \big( \pi_{L_{j-1}+1}+\ldots+\pi_{L_j}\big),
\end{equation*}
where $L_0 \equiv 0$. Hence, the total welfare of all agents is 
\begin{equation*}
\sum_{j=1}^n v_{(j)} \big( \pi_{L_{j-1}+1}+\ldots+\pi_{L_j}\big).
\end{equation*}
For any other valuation profile arbitrarily close to $\mathbf{v}$, agents in 
$\mathbf{v}[j]$ will (a) have valuations arbitrarily close to $v_{(j)}$ and
(b) their ranks (in the valuation profile) will be from $L_{j-1}+1$ to $L_j$%
. As a result, their total welfare is arbitrarily close to 
\begin{equation*}
v_{(j)} \big( \pi_{L_{j-1}+1}+\ldots+\pi_{L_j}\big).
\end{equation*}
Applying this argument to every $j$, we get that for any valuation
arbitrarily close to $\mathbf{v}$, the total welfare of agents is
arbitrarily close to 
\begin{equation*}
\sum_{j=1}^n v_{(j)} \big( \pi_{L_{j-1}+1}+\ldots+\pi_{L_j}\big)=\sum_{i \in
N}v_if_i(\mathbf{v}).
\end{equation*}

Steps 1 and 2 show that $R^f$ is continuous in $\mathbf{v}$.
\end{proof}

\subsection*{Proof of Theorem \protect\ref{theo:char}}

\begin{proof}
Suppose $f$ is a symmetric allocation rule which is satisfactorily
implementable. This implies that there exists a symmetric $\mathbf{p}$ such
that the mechanism $M \equiv (f,\mathbf{p})$ is satisfactory. By Proposition %
\ref{prop:myerson}, $f$ is monotone. The remainder of the claims we do in
steps. \newline

\noindent \textsc{Step 1.} In this step, we show that for every $\mathbf{v}
\in V^n$ such that $N^0_{\mathbf{v}} \ne \emptyset$, we have for every $i
\in N^0_{\mathbf{v}}$, 
\begin{align*}
\mathcal{U}^M_i(\mathbf{v}) &= \frac{1}{|N^0_{\mathbf{v}}|} \sum_{T
\subseteq N: N^0_{\mathbf{v}} \subseteq T} \frac{(-1)^{|T \setminus N^0_{%
\mathbf{v}} |}}{C(|T|,|N^0_{\mathbf{v}}|)} R^f(0_T,v_{-T}).
\end{align*}
We show this using induction. If $|N^0_{\mathbf{v}}| = n$, then
budget-balance implies that $\sum_{i \in N}\mathcal{U}^M_i(\mathbf{v})=0$.
Symmetry implies that $\mathcal{U}^M_j(\mathbf{v})=\mathcal{U}^M_k(\mathbf{v}%
)$ for all $j,k \in N$ at this valuation profile. Hence, $\mathcal{U}^M_i(%
\mathbf{v})=0$ for all $i \in N$. Since $\mathbf{v} \equiv 0_N$, we have $%
R^f(\mathbf{v})=0$. Hence, the claim is true for $N^0=N$.

Suppose the claim is true for all valuation profiles $\mathbf{\bar{v}}$ such
that $|N^0_{\mathbf{\bar{v}}}| > |N^0_{\mathbf{v}}|$. Let $K \equiv N^0_{%
\mathbf{v}}$. Since $M$ is DSIC and budget-balanced, by Proposition \ref%
{prop:myerson}, we get 
\begin{align*}
R^f(\mathbf{v}) &= \sum_{i \in N}\mathcal{U}^M_i(0,v_{-i}) = \sum_{i \in K}%
\mathcal{U}^M_i(0,v_{-i}) + \sum_{i \notin K}\mathcal{U}^M_i(0,v_{-i}) \\
&= \sum_{i \in K}\mathcal{U}^M_i(0_K,v_{-K}) + \sum_{i \notin K}\mathcal{U}%
^M_i(0_{K \cup \{i\}},v_{-(K \cup \{i\})}) \\
&= |K| \mathcal{U}^M_j(0_K,v_{-K}) + \sum_{i \notin K}\mathcal{U}^M_i(0_{K
\cup \{i\}},v_{-(K \cup \{i\})}) \qquad~\text{ (where $j$ is some agent in $%
K $)} \\
&= |K| \mathcal{U}^M_j(0_K,v_{-K}) + \frac{1}{|K|+1} \sum_{i \notin
K}\sum_{T \subseteq N: (K \cup \{i\}) \subseteq T} \frac{(-1)^{|T \setminus
K | - 1}}{C(|T|,(|K|+1))}R^f(0_T,v_{-T}),
\end{align*}
where the third equality followed from symmetry and the final equality
followed from the induction hypothesis. The summation in the last line of
the above sequence of expressions can be simplified as follows: 
\begin{align*}
& \sum_{i \notin K}\sum_{T \subseteq N: (K \cup \{i\}) \subseteq T} \frac{%
(-1)^{|T \setminus K | - 1}}{C(|T|,(|K|+1))}R^f(0_T,v_{-T}) \\
&= \sum_{T \subseteq N: K \subsetneq T} \frac{(-1)^{|T \setminus K | - 1}}{%
C(|T|,(|K|+1))}(|T \setminus K|)~R^f(0_T,v_{-T}) \\
&= \sum_{T \subseteq N: K \subsetneq T}\frac{(-1)^{|T \setminus K | - 1}}{%
C(|T|,|K|)}(|K|+1)~R^f(0_T,v_{-T}).
\end{align*}
To understand why the first equality works, note that for every $T \subseteq
N$ such that $K \subseteq T$, the summation will come for all $i \in T
\setminus K$. Hence, it will appear $(|T \setminus K|)$ times.

Using the above equations in the earlier expression, we get that for all $j
\in K$, 
\begin{align*}
\mathcal{U}^M_j(0_K,v_{-K}) &= \frac{1}{|K|} R^f(\mathbf{v}) + \frac{1}{|K|}
\sum_{T \subseteq N: K \subsetneq T}\frac{(-1)^{|T \setminus K |}}{C(|T|,|K|)%
}~R^f(0_T,v_{-T}) \\
&= \frac{1}{|K|} R^f(0_K,v_{-K}) + \frac{1}{|K|} \sum_{T \subseteq N: K
\subsetneq T}\frac{(-1)^{|T \setminus K |}}{C(|T|,|K|)}~R^f(0_T,v_{-T}) \\
&= \frac{1}{|K|}\sum_{T \subseteq N: K \subseteq T}\frac{(-1)^{|T \setminus
K |}}{C(|T|,|K|)}~R^f(0_T,v_{-T})
\end{align*}
This proves the claim. \newline

\noindent \textsc{Step 2.} Now consider any valuation profile $\mathbf{v}$.
By Proposition \ref{prop:myerson}, we see that for every agent $i \in N$, 
\begin{align*}
p_i(\mathbf{v}) &= R^f_i(\mathbf{v}) - \mathcal{U}^M_i(0,v_{-i}).
\end{align*}
Using Step 1 in this equation gives us the desired expression for $p_i(%
\mathbf{v})$. \newline

\noindent \textsc{Step 3.} Finally, we show that $f$ is residually balanced.
Consider any type profile $\mathbf{v}$ such that $N^0_{\mathbf{v}} =
\emptyset$. Then, using Step 2, for every $i \in N$, 
\begin{align*}
p_i(\mathbf{v}) &= R^f_i(\mathbf{v}) - \sum_{T \subseteq N: i \in T}\frac{%
(-1)^{|T|-1}}{|T|}R^f(0_T,v_{-T}).
\end{align*}
Hence, we get 
\begin{align*}
0 = \sum_{i \in N}p_i(\mathbf{v}) &= R^f(\mathbf{v}) + \sum_{i \in N}\sum_{T
\subseteq N: i \in T}\frac{(-1)^{|T|}}{|T|}R^f(0_T,v_{-T}) \\
&= R^f(\mathbf{v}) + \sum_{T \subseteq N: T \ne
\emptyset}(-1)^{|T|}R^f(0_T,v_{-T}) \\
&= \sum_{T \subseteq N}(-1)^{|T|}R^f(0_T,v_{-T}).
\end{align*}
This shows that $f$ is residually balanced. This concludes one direction of
our proof.

For the other direction, suppose $f$ is a symmetric allocation rule that is
monotone and residually balanced. Consider $\mathbf{p}$ defined in the
statement of this theorem. Clearly, $\mathbf{p}$ is symmetric since $f$ is
symmetric. Hence, $M \equiv (f,\mathbf{p})$ is a symmetric mechanism.
Further, for every agent $i \in N$ and every valuation profile $\mathbf{v}$,
we get 
\begin{equation*}
\mathcal{U}^M_i(v_i,v_{-i}) = v_i f_i(v_i,v_{-i}) - R^f_i(v_i,v_{-i}) + 
\mathcal{U}^M_i(0,v_{-i}),
\end{equation*}
where we have used the expression for $p_i(\mathbf{v})$ to substitute it
with $R^f_i(\mathbf{v}) - \mathcal{U}^M_i(0,v_{-i})$ in the above
expression. This gives us 
\begin{equation*}
\mathcal{U}^M_i(v_i,v_{-i}) = \mathcal{U}^M_i(0,v_{-i}) +
\int_0^{v_i}f_i(x_i,v_{-i})dx_i.
\end{equation*}
This along with the monotonicity of $f$ implies $M$ is DSIC (Proposition \ref%
{prop:myerson}).

Finally, we show that $M$ is budget-balanced. To do so, consider a valuation
profile $\mathbf{v}$. We consider two cases. \newline

\noindent \textsc{Case 1.} $N^0_{\mathbf{v}} \ne \emptyset$. Let $K \equiv
N^0_{\mathbf{v}}$. Now, 
\begin{align*}
\sum_{i \in N}p_i(\mathbf{v}) &= \sum_{i \in K}p_i(\mathbf{v}) + \sum_{i
\notin K}p_i(\mathbf{v}) \\
&= \sum_{i \in K}\bigg[ R^f_i(\mathbf{v}) - \frac{1}{|K|} \sum_{T \subseteq
N: K \subseteq T} \frac{(-1)^{|T \setminus K|}}{C(|T|,|K|)} R^f(0_T,v_{-T}) %
\bigg] \\
&+ \sum_{i \notin K}\bigg[ R^f_i(\mathbf{v}) - \frac{1}{|K|+1} \sum_{T
\subseteq N: (K \cup \{i\}) \subseteq T} \frac{(-1)^{|T \setminus K| - 1}}{%
C(|T|,(|K|+1))}R^f(0_T,v_{-T})\bigg] \\
&= R^f(\mathbf{v}) - \sum_{T \subseteq N: K \subseteq T} \frac{(-1)^{|T
\setminus K|}}{C(|T|,|K|)} R^f(0_T,v_{-T}) \\
&- \sum_{i \notin K}\bigg[\frac{1}{|K|+1} \sum_{T \subseteq N: (K \cup
\{i\}) \subseteq T} \frac{(-1)^{|T \setminus K| - 1}}{C(|T|,(|K|+1))}%
R^f(0_T,v_{-T})\bigg] \\
&= R^f(\mathbf{v}) - \sum_{T \subseteq N: K \subseteq T} \frac{(-1)^{|T
\setminus K|}}{C(|T|,|K|)} R^f(0_T,v_{-T}) \\
&+ \bigg[\sum_{T \subseteq N: K \subsetneq T} \frac{(-1)^{|T \setminus K|}}{%
C(|T|,|K|)}R^f(0_T,v_{-T})\bigg] \\
&= R^f(\mathbf{v}) - R^f(0_K,v_{-K}) \\
&= 0.
\end{align*}

Note that budget-balance followed without any extra conditions in this case.

\noindent \textsc{Case 2.} $N^0_{\mathbf{v}} = \emptyset$. In that case, 
\begin{align*}
\sum_{i \in N}p_i(\mathbf{v}) &= R^f(\mathbf{v}) + \sum_{i \in N}\sum_{T
\subseteq N: i \in T}\frac{(-1)^{|T|}}{|T|}R^f(0_T,v_{-T}) \\
&= R^f(\mathbf{v}) + \sum_{T \subseteq N: T \ne
\emptyset}(-1)^{|T|}~R^f(0_T,v_{-T}) \\
&= \sum_{T \subseteq N}(-1)^{|T|}~R^f(0_T,v_{-T}) \\
&= 0,
\end{align*}
where the last equality follows from the fact that $f$ is residually
balanced. \newline

\noindent This completes the proof.
\end{proof}

\subsection*{Proof of Proposition \protect\ref{prop:rank}}

In this section, we give a proof of Proposition \ref{prop:rank}. We
extensively use Theorem \ref{theo:char} to prove our result. Before starting
our proofs, we explicitly compute the $R^f$ values for any ranking
allocation rule $f$. A valuation profile $\mathbf{v}$ is called \textbf{%
0-generic} if for all $i \ne j$ with $v_i=v_j$ we have $v_i=v_j=0$. \newline

We start off with the following claim.

\begin{lemma}
\label{lem:rcomp} Suppose $f$ is a ranking allocation rule with allocation
probabilities $\pi \equiv (\pi_1,\ldots,\pi_n)$. Then, for every 0-generic
valuation profile $\mathbf{v}$, we have 
\begin{align*}
R^f(\mathbf{v}) &= \sum_{j=1}^{n-1} j v_{(j+1)} (\pi_j - \pi_{j+1}),
\end{align*}
where $v_{(k)}=0$ if $\mathbf{v}[k]=\emptyset$ for any $k$.
\end{lemma}

\begin{proof}
Choose a 0-generic valuation profile $\mathbf{v}$. Consider agent $i \in N$
with $v_i > 0$. Since $\mathbf{v}$ is a 0-generic valuation profile, $\{i\}
= \mathbf{v}[j]$ for some $j$. If $j=n$, then $v_if_i(\mathbf{v}) -
\int_0^{v_i}f_i(x_i,v_{-i})dx_i = 0$. So, consider $j < n$. As a result 
\begin{align*}
v_if_i(\mathbf{v}) - \int_0^{v_i}f_i(x_i,v_{-i})dx_i &= \pi_jv_{(j)} -
\int_0^{v_{(j)}}f_i(x_i,v_{-i})dx_i \\
&= \pi_jv_{(j)} - \sum_{h = j}^n \pi_h (v_{(h)} - v_{(h+1)}) \qquad~\text{%
(Note: $v_{(n+1)} \equiv 0$.)} \\
&= \sum_{h = j+1}^n v_{(h)} (\pi_{h-1} - \pi_h).
\end{align*}
This implies that 
\begin{align*}
R^f(\mathbf{v}) &= \sum_{j=1}^{n-1} \sum_{h=j+1}^n v_{(h)}(\pi_{h-1} - \pi_h)
\\
&= \sum_{j=1}^{n-1} j v_{(j+1)} (\pi_j - \pi_{j+1}).
\end{align*}
\end{proof}

Using Lemma \ref{lem:rcomp}, we will now give a proof of Proposition \ref%
{prop:rank}. \newline

\noindent \textsc{Proof of Proposition \ref{prop:rank}.} \newline

\begin{proof}
Let $f$ be a ranking allocation rule with allocation probabilities $%
(\pi_1,\ldots,\pi_n)$. Note that $f$ is monotone in the sense of Myerson. By
Theorem \ref{theo:char}, we know that $f$ is satisfactorily implementable if
and only if for every $\mathbf{v}$ with $v_1 \ge v_2 \ge \ldots \ge v_n > 0$%
, we have 
\begin{align*}
\sum_{T \subseteq N}(-1)^{|T|}R^f(0_T,v_{-T}) = \sum_{k=0}^n \sum_{T
\subseteq N: |T|=n-k}(-1)^{n-k} R^f(0_T,v_{-T}) = 0.
\end{align*}
Since $R^f$ is continuous (Lemma \ref{lem:contin}), it is enough to show the
above equality for 0-generic valuation profiles. In other words, continuity
of $R^f$ implies that $f$ is satisfactorily implementable if and only if for
every $\mathbf{v}$ with $v_1 > v_2 > \ldots > v_n > 0$, we have 
\begin{align*}
\sum_{T \subseteq N}(-1)^{|T|}R^f(0_T,v_{-T}) = \sum_{k=0}^n \sum_{T
\subseteq N: |T|=k}(-1)^k R^f(0_T,v_{-T}) = 0.
\end{align*}
Note that for every $T \subseteq N$, the profile $(0_T,v_{-T})$ is a
0-generic valuation profile. We can divide this sum into two parts. 
\begin{align*}
\sum_{T \subseteq N}(-1)^{|T|}R^f(0_T,v_{-T}) = \sum_{T \subseteq N: n \in
T}(-1)^{|T|}R^f(0_T,v_{-T}) + \sum_{T \subseteq N: n \notin
T}(-1)^{|T|}R^f(0_T,v_{-T})
\end{align*}
Hence, we can write the residual balancedness condition as 
\begin{align*}
\sum_{T \subseteq N}(-1)^{|T|}R^f(0_T,v_{-T}) = \sum_{T \subseteq N: n
\notin T} (-1)^{|T|}\Big[ R^f(0_T,v_{-T}) - R^f(0_{T \cup \{n\}},v_{-(T \cup
\{n\})})\Big] = 0.
\end{align*}

Now, fix a $T \subseteq N$ with $n \notin T$ and $|T|=n-k$. Since $\mathbf{v}
$ is a 0-generic valuation profile, the rank of agent $n$ in $(0_T,v_{-T})$
is $k$. Without loss of generality, we denote $(0_T,v_{-T}) \equiv \mathbf{v}%
^{\prime }$. Note that $v^{\prime }_{(k)}=v_n$. Using Lemma \ref{lem:rcomp}, 
\begin{equation*}
R^f(0_T,v_{-T}) = \sum_{j=1}^{k-1} jv^{\prime }_{(j+1)}(\pi_j - \pi_{j+1})
\end{equation*}
and 
\begin{equation*}
R^f(0_{(T \cup \{n\})},v_{-(T \cup \{n\})}) = \sum_{j=1}^{k-2} jv^{\prime
}_{(j+1)}(\pi_j - \pi_{j+1}).
\end{equation*}
Hence, we can write 
\begin{equation*}
R^f(0_T,v_{-T}) - R^f(0_{T \cup \{n\}},v_{-(T \cup \{n\})}) = (k-1)v^{\prime
}_{(k)} (\pi_{k-1} - \pi_k) = (k-1)v_n (\pi_{k-1} - \pi_k),
\end{equation*}
where the last equality follows because $v^{\prime }_{(k)}=v_n$. Note that
the RHS only depends on the size of $T$ but not on which elements are
present in $T$. As a result, we can write the residual balancedness
condition as 
\begin{align*}
0 = \sum_{T \subseteq N}(-1)^{|T|}R^f(0_T,v_{-T}) &= \sum_{T \subseteq N: n
\notin T} (-1)^{|T|}\Big[ R^f(0_T,v_{-T}) - R^f(0_{T \cup \{n\}},v_{-(T \cup
\{n\})})\Big] \\
&= \sum_{k=1}^{n} \sum_{T \subseteq N: n \notin T, |T|=n-k}(-1)^{n-k}
(k-1)v_n (\pi_{k-1} - \pi_k) \\
&= \sum_{k=2}^{n} (-1)^{n-k} C(n-1,k-1) (k-1)(\pi_{k-1} - \pi_k)v_n.
\end{align*}
The last inequality follows because we can form a subset of size $n-k$ from $%
n-1$ elements in $C(n-1,k-1)$ ways. Now, we simplify this expression to get
our desired result. Since $v_n > 0$, residual balancedness is equivalent to: 
\begin{align*}
0 &= \sum_{k=2}^{n} (-1)^{n-k} C(n-1,k-1) (k-1)(\pi_{k-1} - \pi_k) \\
&= \sum_{k=2}^{n} (-1)^{n-k} C(n-1,k-1) (k-1)\pi_{k-1} - \sum_{k=2}^{n}
(-1)^{n-k} C(n-1,k-1) (k-1)\pi_{k} \\
&= -\sum_{\ell=1}^{n-1}(-1)^{n-\ell}C(n-1,\ell)\ell \pi_{\ell} -
\sum_{\ell=1}^{n} (-1)^{n-\ell} C(n-1,\ell-1) (\ell-1)\pi_{\ell} \\
&= -\sum_{\ell=1}^{n-1}(-1)^{n-\ell}\pi_{\ell} \big[\ell C(n-1,\ell) +
(\ell-1)C(n-1,\ell-1) \big] - (-1)^0 (n-1)C(n-1,n-1)\pi_n \\
&= -\sum_{\ell=1}^{n-1}(-1)^{n-\ell}(n-1)C(n-1,\ell-1)\pi_{\ell} - (-1)^0
(n-1)C(n-1,n-1)\pi_n \\
& \text{(Here, we used the fact that $\ell C(n-1,\ell) +
(\ell-1)C(n-1,\ell-1)=(n-1)C(n-1,\ell-1)$.)} \\
&= -\sum_{\ell=1}^n (-1)^{n-\ell}(n-1)C(n-1,\ell-1)\pi_{\ell}
\end{align*}
Since $n > 1$, we get that residual balancedness is equivalent to 
\begin{equation*}
0=\sum_{\ell=1}^n (-1)^{n-\ell}C(n-1,\ell-1)\pi_{\ell}.
\end{equation*}
This can be equivalently written as 
\begin{equation*}
0=\sum_{\ell=1}^n (-1)^{\ell}C(n-1,\ell-1)\pi_{\ell},
\end{equation*}
which is the desired claim.
\end{proof}

\subsection*{Proofs of Theorem \protect\ref{theo:nrgl}}

In this section, we give a proof of Theorem \ref{theo:nrgl}. We start by
characterizing the two-step ranking allocation rules that can be
satisfactorily implemented.

\begin{prop}
\label{prop:2rank} A two-step ranking allocation rule is satisfactorily
implementable if and only if $2 \le \ell \le n-1$, $\ell$ is even, and 
\begin{equation*}
\pi_1 = \frac{C(n-2,\ell-1)+1}{C(n-2,\ell-1)+\ell}.
\end{equation*}
\end{prop}

\begin{proof}
In this and subsequent proofs, we use the following combinatorial fact.

\begin{fact}
\label{fact:f1} For any $r \in \{0,\ldots,n-1\}$, 
\begin{equation*}
\sum_{j=0}^r(-1)^j C(n,j) = (-1)^r C(n-1,r)
\end{equation*}
and 
\begin{equation*}
\sum_{j=0}^n (-1)^j C(n,j) = 0.
\end{equation*}
\end{fact}

By Proposition \ref{prop:rank}, we know that for any two-step ranking
allocation rule defined by $(\pi_1,\ell)$, satisfactorily implementability
is equivalent to 
\begin{align}  \label{eq:ir00}
-\pi_1 + \sum_{k=2}^{\ell}(-1)^k C(n-1,k-1) \pi_2 &= 0.
\end{align}
This immediately implies that $\ell \ne 1$. Further, if $\ell=n$, then we
must have $\pi_1= \sum_{k=2}^{n}(-1)^k C(n-1,k-1) \pi_2 = \pi_2$. But, by
definition of a two-step allocation rule $\pi_1 > \pi_2$. So, we have $1 <
\ell < n$.

Now, using Fact \ref{fact:f1}, 
\begin{align*}
\sum_{k=2}^{\ell}(-1)^k C(n-1,k-1) &= -\sum_{k=1}^{\ell-1}(-1)^k C(n-1,k) \\
&=1 - \Big[\sum_{k=0}^{\ell-1}(-1)^k C(n-1,k) \Big] \\
&= 1-(-1)^{\ell-1} C(n-2,\ell-1) \\
&= 1+ (-1)^{\ell} C(n-2,\ell-1).
\end{align*}
Using this and the fact that $\pi_2=\frac{1}{\ell-1} (1-\pi_1)$, we simplify
Equation \ref{eq:ir00} as 
\begin{align*}
-\pi_1 + \frac{1}{\ell-1} (1-\pi_1) \Big( 1+ (-1)^{\ell} C(n-2,\ell-1) \Big) %
&= 0.
\end{align*}
For this to hold, we must have $\ell$ even and 
\begin{equation*}
\pi_1= \frac{C(n-2,\ell-1)+1}{C(n-2,\ell-1)+\ell}.
\end{equation*}
\end{proof}

We now provide a proof of Theorem \ref{theo:nrgl}. \newline

\noindent \textsc{Proof of Theorem \ref{theo:nrgl}.} \newline

\begin{proof}
We do the proof in several steps. \newline

\noindent \textsc{Step 1 - The Primal Problem.} In this step, we formulate
the problem of finding an r-optimal allocation rule as a linear program.

Pick $\epsilon \in \mathbb{R}^n$ sufficiently close to the zero vector. Note
that $\epsilon$ may be the $n$-dimensional zero vector or a vector with
negative components. We formulate a linear program (in terms of $\epsilon$)
as follows.

\begin{align*}
\max_{(\pi_1,\ldots,\pi_n)} \pi_1 + \sum_{j=1}^n \epsilon_j \pi_j & ~~~ \\
\text{s.t.}~~~~& ~~~~~~~~~~~~\mathbf{(LP-RANK)} \\
\pi_{i+1} - \pi_i &\le 0~~~\qquad~\forall~i \in \{1,\ldots,n-1\} \\
\sum_{i=1}^n (-1)^i C(n-1,i-1) \pi_i &= 0 \\
\sum_{i=1}^n \pi_i &= 1 \\
\pi_i &\ge 0~~~\qquad~\forall~i \in \{1,\ldots,n\}.
\end{align*}

By Proposition \ref{prop:rank}, a feasible solution to the linear program 
\textbf{(LP-RANK)} is a satisfactorily implementable ranking allocation
rule. Note that we have imposed $\sum_{i=1}^n \pi_i=1$ instead of weak
inequality. Since we are interested in finding an r-optimal allocation rule,
by Lemma \ref{lem:add1}, this is without loss of generality. Also, if $%
\epsilon$ is the zero vector, then the optimal solution of this linear
program will give us an r-optimal allocation rule. We will find an optimal
solution of \textbf{(LP-RANK)} for all $\epsilon$ sufficiently close to the
zero vector. This will ensure that such an optimal solution is the unique
r-optimal allocation rule. \newline

\noindent \textsc{Step 2 - The Dual Problem} We first consider the dual of 
\textbf{(LP-RANK)} and construct a dual feasible solution. For formulating
the dual, we associate a variable $\theta_i$ for each of the constraint in
the first set of constraints corresponding to $i \in \{1,\ldots,n-1\}$. We
also associate variables $y$ and $z$ for the second and third constraints
respectively.

This leads us to the dual of the linear program \textbf{(LP-RANK)}. 
\begin{align*}
\min_{(y,z,(\theta_1,\ldots,\theta_{n-1}))} z & ~~~~ \\
\text{s.t.}~~~~~~&~~~~~~\mathbf{(DP-RANK)} \\
-\theta_1 - y + z &\ge 1+\epsilon_1 \\
\theta_{i-1} - \theta_i + (-1)^i C(n-1,i-1)y + z &\ge \epsilon_i
\qquad~\forall~i \in \{2,\ldots,n-1\} \\
\theta_{n-1} + (-1)^{n} y + z &\ge \epsilon_n \\
\theta_i &\ge 0~\qquad~\forall~i \in \{1,\ldots,n-1\}.
\end{align*}

We construct a dual feasible solution as follows. Set $\theta_1=0$ and we
will choose $y$ and $z$ such that $z-y=1+\epsilon_1$. This will imply that
the first constraint is automatically satisfied. The rest of the constraints
are satisfied by successively computing $\theta_i$ for $i \in
\{2,\ldots,n-1\}$. First, we set 
\begin{equation*}
\theta_2=\theta_1 + (-1)^2 C(n-1,1)y + z - \epsilon_2 = (-1)^2 C(n-1,1)y + z
- \epsilon_2 .
\end{equation*}
Then, 
\begin{equation*}
\theta_3=\theta_2 + (-1)^3 C(n-1,2)y + z - \epsilon_3 = \Big( (-1)^2
C(n-1,1)+(-1)^3 C(n-1,2) \Big)y + 2z - \epsilon_2 - \epsilon_3.
\end{equation*}
Continuing in this manner, we have for all $i \in \{2,\ldots,n-1\}$, 
\begin{align*}
\theta_i &= \Big( \sum_{j=1}^{i-1} (-1)^{j+1} C(n-1,j) \Big) y + (i-1)z -
\sum_{j=2}^i \epsilon_j \\
&= (i-1)z - \sum_{j=2}^i \epsilon_j - \Big( \sum_{j=1}^{i-1}(-1)^j C(n-1,j) %
\Big) y \\
&= (i-1)z - \sum_{j=2}^i \epsilon_j - \Big( (-1)^{i-1}C(n-2,i-1) - 1 \Big) y
\qquad~\text{(Using Fact \ref{fact:f1})} \\
&= (i-1)z - \sum_{j=2}^i \epsilon_j - \Big( (-1)^{i-1}C(n-2,i-1) - 1 \Big) %
(z-1) \\
&+ \epsilon_1 \Big( (-1)^{i-1}C(n-2,i-1) - 1 \Big) \qquad~\text{(Using $y=z
- 1 - \epsilon_1$)} \\
&= (1+\epsilon_1) \Big((-1)^{i-1}C(n-2,i-1) - 1 \Big) - z\Big( %
(-1)^{i-1}C(n-2,i-1) - i \Big) - \sum_{j=2}^i \epsilon_j.
\end{align*}
This choice of $\theta_i$ ensures that the second set of inequalities in 
\textbf{DP-RANK} are satisfied. However, we need to make sure that (a) $%
\theta_i$s are non-negative and (b) the last inequality is satisfied. These
are ensured by choosing $y$ and $z$ appropriately.

For every $i \in \{2,\ldots,n-1\}$, let 
\begin{equation*}
H(n,i) := (-1)^{i-1}C(n-2,i-1).
\end{equation*}

First, for non-negativity of $\theta_i$, we will choose $z$ appropriately.
Note that $(1+\epsilon_1) > 0$ since $\epsilon_1$ is sufficiently close to
zero. Further, $\theta_i \ge 0$ if and only if 
\begin{align}  \label{eq:thno}
(1+\epsilon_1) \Big(H(n,i) - 1 \Big) - z \Big(H(n,i) - i \Big) -
\sum_{j=2}^i \epsilon_j &\ge 0.
\end{align}
We consider two cases. \newline

\noindent \textsc{Case a.} If $i$ is even, we have $H(n,i)=-C(n-2,i-1) < 0$.
Simplifying, we get 
\begin{align}
z &\ge (1+\epsilon_1) \frac{\Big(C(n-2,i-1)+1 \Big)}{\Big(C(n-2,i-1) + i %
\Big)} +\frac{1}{\Big(C(n-2,i-1) + i \Big)} \sum_{j=2}^i \epsilon_j  \notag
\\
&= (1+\epsilon_1) \Bigg( 1 - \frac{(i-1)}{\Big(C(n-2,i-1) + i \Big)} \Bigg) %
+ \frac{1}{\Big(C(n-2,i-1) + i \Big)} \sum_{j=2}^i \epsilon_j.
\label{eq:zmin}
\end{align}
Note that if $i=2$, we need 
\begin{equation*}
z \ge (1+\epsilon_1) (1-\frac{1}{n}) + \frac{1}{n} \epsilon_2
\end{equation*}
Now, choose $\ell$ as follows: 
\begin{align}  \label{eq:ell}
\ell \in \arg \min_{n-1 \ge i \ge 2, i~\text{even}} \frac{(i-1)}{\Big(%
C(n-2,i-1) + i \Big)}.
\end{align}
Observe that as $\epsilon$ is sufficiently close to the zero vector, the
second term on the RHS of Inequality \ref{eq:zmin} is very small (close to
zero) for all $i$. Hence, this choice of $\ell$ maximizes the RHS of
Inequality \ref{eq:zmin} if (a) $\epsilon$ is the zero vector or (b) there
is a unique $\ell$ that minimizes the expression in (\ref{eq:ell}) - if
there are more than one $\ell$ which minimizes the expression in (\ref%
{eq:ell}), then the RHS of Inequality (\ref{eq:zmin}) is minimized by
looking at the second term. By Corollary \ref{cor:comp}, if $n \ne 8$, then
there is a unique $\ell$ that minimizes the expression in (\ref{eq:ell}).
For $n = 8$, there are two possible values of $\ell$ that minimize this the
expression in (\ref{eq:ell}). As a result, which choice of $\ell$ maximizes
the RHS of Inequality \ref{eq:zmin} will depend on the value of $\epsilon$ -
if $\epsilon$ is the zero vector, then either choice will work.

This implies that for $\epsilon$ sufficiently close to the zero vector and $%
n \ne 8$, Inequality \ref{eq:zmin} can be satisfied by choosing $z=z^*$,
where 
\begin{equation*}
z^*:= (1+\epsilon_1) \Bigg( 1 - \frac{(\ell-1)}{\Big(C(n-2,\ell-1) + \ell %
\Big)} \Bigg) + \frac{1}{\Big(C(n-2,\ell-1) + \ell \Big)} \sum_{j=2}^i
\epsilon_j.
\end{equation*}
For $n=8$, choice of $z=z^*$, where $z^*$ is defined by choosing any $\ell$
that minimizes the expression in (\ref{eq:ell}), satisfies Inequality \ref%
{eq:zmin} if $\epsilon$ is the zero vector.

As argued earlier, $z^* \ge (1+\epsilon_1) (1-\frac{1}{n}) + \frac{1}{n}
\epsilon_2$. \newline

\noindent\textsc{Case b.} If $i$ is odd, then $H(n,i)=C(n-2,i-1)$. If $i=n-1$%
, then Inequality \ref{eq:thno} reduces to $z (n-2) - \sum_{j=2}^{n-1}
\epsilon_j \ge 0$. Since $n \ge 3$ and $\epsilon$ is sufficiently close to
the zero vector, by choosing $z=z^*$, it is satisfied. Hence, we assume $i <
n-1$. In that case $H(n,i) \ge i$. If $H(n,i) - i=0$, then the desired
Inequality (\ref{eq:thno}) is satisfied for any choice of $z$ since $%
\epsilon $ is sufficiently close to the zero vector. Assume that $H(n,i) > i$%
. Then, Inequality \ref{eq:thno} holds if 
\begin{align*}
z &\le (1+\epsilon_1) \frac{\Big(C(n-2,i-1) - 1\Big)}{\Big(C(n-2,i-1)-i %
\Big) } - \frac{1}{\Big(C(n-2,i-1)-i \Big)} \sum_{j=2}^i\epsilon_j \\
&= (1+\epsilon_1) \Bigg( 1+ \frac{(i-1)}{\Big(C(n-2,i-1)-i\Big)} \Bigg) - 
\frac{1}{\Big(C(n-2,i-1)-i \Big)} \sum_{j=2}^i\epsilon_j.
\end{align*}
Since $\epsilon$ is arbitrarily close to the zero vector, by setting $z=z^*$%
, this inequality is trivially satisfied. \newline

Hence, we choose $z=z^*$, where 
\begin{align}  \label{eq:expz}
z^* &= (1+\epsilon_1) \Bigg( 1 - \frac{(\ell-1)}{\Big(C(n-2,\ell-1) + \ell %
\Big)} \Bigg) + \frac{1}{\Big(C(n-2,\ell-1) + \ell \Big)} \sum_{j=2}^i
\epsilon_j.
\end{align}
Hence, we satisfy the non-negativity constraints by this choice of $z$. Let $%
y^*=z^*-1 - \epsilon_1$. Finally, we show that the last inequality in 
\textbf{DP-RANK} is satisfied. To see this, if $n$ is odd, then the
inequality reduces to 
\begin{equation*}
\theta_{n-1} - y^* + z^* = \theta_{n-1} + 1 + \epsilon_1 \ge \epsilon_n,
\end{equation*}
where the inequality follows since we have chosen $\theta_{n-1} \ge 0$ and $%
\epsilon$ is arbitrarily close to the zero vector.

If $n$ is even, we note that $\theta_{n-1}=z(n-2) -
\sum_{j=2}^{n-1}\epsilon_j$ by definition. Then the inequality reduces to 
\begin{align*}
\theta_{n-1} + y^* + z^* &= z^*(n-2) + 2z^* -1 - \sum_{j=1}^{n-1}\epsilon_j
\\
&= z^*n - 1 - \sum_{j=1}^{n-1}\epsilon_j \\
&\ge n (1+\epsilon_1) (1-\frac{1}{n}) + \epsilon_2 - 1 -
\sum_{j=1}^{n-1}\epsilon_j \\
&= n -2 + \epsilon_1(n-1) + \epsilon_2 - \sum_{j=1}^{n-1}\epsilon_j \\
&\ge \epsilon_n,
\end{align*}
where we used the fact that $z^* \ge (1+\epsilon_1) (1-\frac{1}{n}) + \frac{1%
}{n} \epsilon_2$, $n \ge 3$, and $\epsilon$ is sufficiently close to the
zero vector in the above inequalities.

This completes the proof that there is a feasible solution of \textbf{%
(DP-RANK)} with $z^*$ defined by Equation \ref{eq:expz}. \newline

\noindent \textsc{Step 3 - Optimality.} In this step, we construct a
feasible solution of \textbf{(LP-RANK)} by constructing the probabilities of
a two-step ranking allocation rule as follows: 
\begin{align*}
\pi_1^* &= 1 - \frac{\ell-1}{\Big(C(n-2,\ell-1) + \ell \Big)}. \\
\pi_i^* &= \frac{1}{\Big(C(n-2,\ell-1) + \ell \Big)} \qquad~\forall~i \in
\{2,\ldots,\ell\} \\
\pi_i^* &= 0 \qquad~\forall~i \in \{\ell+1,\ldots,n\}.
\end{align*}
By construction, $\sum_{j \in N}\pi^*_j=1$ and $\pi^*_1 \ge \pi^*_i$ for all 
$i \in \{2,\ldots,\ell\}$. By Proposition \ref{prop:2rank}, $%
(\pi_1^*,\ldots,\pi_n^*)$ is a feasible solution of \textbf{(LP-RANK)}.
Further, we see that the objective function value of \textbf{(LP-RANK)} with
this feasible solution is 
\begin{equation*}
(1+\epsilon_1)\Bigg( 1 - \frac{\ell-1}{\Big(C(n-2,\ell-1) + \ell \Big)} %
\Bigg) + \sum_{j=2}^{\ell} \epsilon_j \frac{1}{\Big(C(n-2,\ell-1) + \ell %
\Big)} = z^*,
\end{equation*}
which is the objective function value of \textbf{(DP-RANK)} for the dual
feasible solution we found in Step 2. Hence, by the strong duality theorem
of linear programming, $(\pi_1^*,\ldots,\pi_n^*)$ is an optimal solution of 
\textbf{(LP-RANK)}. For all $n \ge 3$, this is an optimal solution when $%
\epsilon$ is the zero vector. Hence, it describes an r-optimal allocation
rule. For $n \ne 8$, this is an optimal solution for all $\epsilon$
arbitrarily close to the zero vector, and hence, it is the unique optimal
solution when $\epsilon$ is \emph{equal} to the zero vector - this follows
from a result by \citet{Mang79}, who showed that an optimal solution of a
linear program is unique if and only if it remains the optimal solution for
sufficiently small perturbation of the objective function.
\end{proof}

\subsection*{Proof of Theorem \protect\ref{theo:ir}}

In this section, we provide a proof of individual rationality of a class of
two-step ranking mechanisms. First, we remind the following elementary fact
from \citet{Myerson81}.

\begin{fact}
\label{fact:ir} A mechanism $(f,\mathbf{p})$ is ex-post individually rational if and
only if for every $i \in N$ and for every $v_{-i}$, we have $p_i(0,v_{-i})
\le 0$.
\end{fact}

Note that the above fact is a \emph{necessary and sufficient} condition for
IR. We now present two useful lemmas that will help us prove Theorem \ref%
{theo:ir}.

\begin{lemma}
\label{lem:ir0} Suppose $f$ is a satisfactorily implementable two-step
ranking allocation rule defined by $(\pi_1,\ell)$. Then, for every 0-generic
valuation profile $\mathbf{v}$, we have 
\begin{equation*}
R^f(\mathbf{v})= (\pi_1-\pi_2) v_{(2)} + \ell \pi_2 v_{(\ell+1)},
\end{equation*}
where $\pi_2=\frac{1}{\ell-1} (1-\pi_1)$.
\end{lemma}

\begin{proof}
The proof of the formula for $R^f$ follows from the formula derived for any
satisfactorily implementable ranking allocation rule in Lemma \ref{lem:rcomp}%
.
\end{proof}

\noindent \textsc{Notation.} For any two positive integers $K, K^{\prime }$
with $K \ge K^{\prime }$, we denote the consecutive product of integers from 
$K^{\prime }$ to $K$ as 
\begin{equation*}
\psi(K^{\prime },K) = K^{\prime }\times (K^{\prime }+1) \times \cdots \times
K.
\end{equation*}

\begin{lemma}
\label{lem:ir1} Suppose $(f,\mathbf{p})$ is a satisfactory mechanism, where $%
f$ is a two-step ranking allocation rule defined by $(\pi_1,\ell)$. Then,
for every $\mathbf{v}$ with $|N^0_{\mathbf{v}}|=n-K$, $K \le \ell$, and $v_1
> \ldots > v_K > 0$, we have for every $i \in N^0_{\mathbf{v}}$, 
\begin{align*}
p_i(\mathbf{v}) &= -\frac{(\pi_1-\pi_2)}{\psi(n-K,n-2)} \Big[%
\sum_{j=2}^{K-1} (-1)^j (j-1)! \psi(n-K,n-j-1)v_j + (-1)^K (K-1)!v_K \Big],~%
\text{if $K \ge 2$},
\end{align*}
and $p_i(\mathbf{v})= 0$ if $K \in \{0,1\}$.
\end{lemma}

\begin{proof}
Pick a satisfactory mechanism $(f,\mathbf{p})$, where $f$ is a two-step
ranking allocation rule defined by $(\pi_1,\ell)$. Suppose $\mathbf{v}$ is
such that $|N^0_{\mathbf{v}}|=n-K$, $K \le \ell$. If $K=0$, then by symmetry
and budget-balance, we get $p_i(\mathbf{v})=0$ for all $i \in N$. Else,
suppose $v_1 > \ldots > v_K > 0$. If $K=1$, then, by budget-balance and
symmetry we get $p_1(\mathbf{v}) + (n-1)p_i(\mathbf{v})=0$ for any $i \in
N^0_{\mathbf{v}}$. But $p_1(\mathbf{v})=p_1(0,v_{-1})+v_1 \pi_1 - v_1 \pi_1
= p_1(0,v_{-1})=0$, where we used revenue equivalence formula for the first
equality and $p_1(0,v_{-1})=0$ for the last equality. Hence, we get $p_1(%
\mathbf{v})=0$, and hence, $p_i(\mathbf{v})=0$ for all $i \ne 1$. Now,
suppose $K=2$. Then, budget-balance requires 
\begin{equation*}
p_1(\mathbf{v})+p_2(\mathbf{v})+\sum_{i \notin \{1,2\}}p_i(\mathbf{v})=0.
\end{equation*}
But using revenue equivalence and the fact that $p_1(0,v_{-1})=0$, we get
that 
\begin{equation*}
p_1(\mathbf{v})=p_1(0,v_{-1})+v_1\pi_1 - (v_1-v_2)\pi_1 - v_2
\pi_2=v_2(\pi_1-\pi_2).
\end{equation*}
Similarly, we get $p_2(\mathbf{v})=p_2(0,v_{-2}) + v_2 \pi_2 - v_2 \pi_2 = 0$%
. Hence, by choosing some $i \notin \{1,2\}$, we can simplify the
budget-balance equation as $v_2 (\pi_1-\pi_2) + (n-2)p_i(\mathbf{v})=0$.
This implies that 
\begin{equation*}
p_i(\mathbf{v}) = - \frac{(\pi_1-\pi_2)}{(n-2)}v_2,
\end{equation*}
which is the required expression.

Next, suppose $K > 2$ and use induction. Suppose the claim is true for all $%
k < K$. Then, by revenue equivalence and symmetry we get 
\begin{equation*}
\sum_{j \in N}p_j(\mathbf{v}) = \sum_{j \in N}p_j(0,v_{-j}) + R^f(\mathbf{v}%
) = (n-K)p_i(\mathbf{v})+ \sum_{j=1}^K p_j(0,v_{-j}) + R^f(\mathbf{v}),
\end{equation*}
where $i$ is some agent in $N^0_{\mathbf{v}}$. By budget-balance, the above
summation is zero, and $R^f(\mathbf{v})=(\pi_1-\pi_2)v_2$ since $K \le \ell$
(by Lemma \ref{lem:ir0}). Using this, we get 
\begin{align}  \label{eq:ire1}
0 &= (n-K)p_i(\mathbf{v})+ \sum_{j=1}^K p_j(0,v_{-j}) + (\pi_1-\pi_2)v_2.
\end{align}
Now, for every $j \in \{1,\ldots,K\}$, the profile $(0,v_{-j})$ has one more
zero-valued agent than the profile $\mathbf{v}$, and hence, we can apply our
induction hypothesis. We refer to $(0,v_{-j})$ for each $j \in
\{1,\ldots,K\} $ as a \textbf{marginal} profile having an additional
zero-valuation agent than $\mathbf{v}$, and denote this as $\mathbf{v}^j$
with the valuation of the $k$-th ranked agent in this valuation profile
denoted as $v^j_{(k)}$. Note that a marginal profile contains $(K-1)$
non-zero valuation agents. Thus, using our induction hypothesis, Equation %
\ref{eq:ire1} can be rewritten as 
\begin{align*}
& (n-K)p_i(\mathbf{v}) \\
&= \sum_{j=1}^{K} \frac{(\pi_1-\pi_2)}{\psi(n-K+1,n-2)} \Big[%
\sum_{k=2}^{K-2} (-1)^k (k-1)! \psi(n-K+1,n-k-1)v^j_{(k)} + (-1)^{K-1}
(K-2)!v^j_{(K-1)} \Big] \\
&- (\pi_1-\pi_2)v_2 \\
&= \frac{(\pi_1-\pi_2)}{\psi(n-K+1,n-2)} \sum_{j=1}^{K}\Big[ %
\sum_{k=2}^{K-2} (-1)^k (k-1)! \psi(n-K+1,n-k-1)v^j_{(k)} + (-1)^{K-1}
(K-2)!v^j_{(K-1)} \Big] \\
&- (\pi_1-\pi_2)v_2
\end{align*}

We write this equivalently as 
\begin{align}  \label{eq:iir}
\frac{\psi(n-K,n-2)}{\pi_1-\pi_2}p_i(\mathbf{v}) &= \sum_{j=1}^{K}\Big[ %
\sum_{k=2}^{K-2} (-1)^k (k-1)! \psi(n-K+1,n-k-1)v^j_{(k)} + (-1)^{K-1}
(K-2)!v^j_{(K-1)} \Big]  \notag \\
&- \psi(n-K+1,n-2)v_2.
\end{align}
Now, we remind that $\mathbf{v}$ is a valuation profile of the form $v_1 >
v_2 > \ldots > v_K > 0$ and $v_j=0$ for all $j > K$. We now simplify the RHS
of Equation \ref{eq:iir} in terms of $v_1,\ldots,v_K$. To do so, we
explicitly compute the coefficients of $v_k$ for each $k \in \{1,\ldots,K\}$
in the RHS of Equation \ref{eq:iir}. \newline

\noindent \textsc{Case 1.} Note that $v_1$ does not appear in the summation,
and hence, its coefficient is always zero. Next, $v_2 = v^j_{(2)}$ for all $%
j \ne \{1,2\}$. Hence, it has a rank $2$ in $(K-2)$ marginal profiles, and
in each such case, its coefficient in the first summation is 
\begin{equation*}
(-1)^2 (1)! \psi(n-K+1,n-3).
\end{equation*}
Adding this with $-\psi(n-K+1,n-2)v_2$, we get the coefficient of $v_2$ as 
\begin{align*}
(K-2) \psi(n-K+1,n-3) - \psi(n-K+1,n-2) = -\psi(n-K,n-3) = -(-1)^2 (1!)
\psi(n-K,n-3).
\end{align*}

\noindent \textsc{Case 2.} Now, consider $K > k > 2$. Note that $v_k =
v^j_{(k^{\prime })}$ where $k^{\prime }\in \{k,k-1\}$. In particular, $%
k^{\prime }=k$ if $j \in \{k+1,\ldots,K\}$ and $k^{\prime }=k-1$ if $j \in
\{1,\ldots,k-1\}$. Hence, it has rank $k$ in $(K-k)$ marginal profiles and
rank $(k-1)$ in $(k-1)$ marginal profiles. When it has rank $k$ in a
marginal profiles, its coefficient in the RHS of Equation \ref{eq:iir} is 
\begin{equation*}
(-1)^k (k-1)! \psi(n-K+1,n-k-1),
\end{equation*}
and when it has rank $(k-1)$, its coefficient is 
\begin{equation*}
(-1)^{k-1} (k-2)! \psi(n-K+1,n-k).
\end{equation*}
Hence, collecting the coefficient of $v_k$, we get 
\begin{align*}
& (-1)^k(K-k) (k-1)! \psi(n-K+1,n-k-1) + (-1)^{k-1} (k-1)(k-2)!
\psi(n-K+1,n-k) \\
&= (-1)^k (k-1)! \psi(n-K+1,n-k-1) \Big( (K-k) - (n-k) \Big) \\
&= - (-1)^k (k-1)! \psi(n-K,n-k-1).
\end{align*}

\noindent \textsc{Case 3.} Finally, $v_K = v^j{(k^{\prime })}$ where $%
k^{\prime }=K-1$ when $j \in \{1,\ldots,K-1\}$. Hence, $v_K$ has rank $(K-1)$
in $(K-1)$ marginal profiles. Whenever it has rank $(K-1)$ its coefficient
in the RHS of Equation \ref{eq:iir} is $(-1)^{K-1}(K-2)!$. Hence, the
coefficient of $v_K$ in the RHS of Equation \ref{eq:iir} is 
\begin{equation*}
-(-1)^K (K-1)(K-2)!=-(-1)^K (K-1)!
\end{equation*}

\noindent Aggregating the findings from all the three cases, we can rewrite
Equation \ref{eq:iir} as 
\begin{align}  \label{eq:ire22}
\frac{\psi(n-K,n-2)}{\pi_1-\pi_2}p_i(\mathbf{v}) &= \Big[%
\sum_{k=2}^{K-1}(-1)^k (k-1)! \psi(n-K,n-k-1)v_k + (-1)^K (K-1)! v_K \Big].
\end{align}
This simplifies to the desire expression: 
\begin{equation*}
p_i(\mathbf{v}) = -\frac{(\pi_1-\pi_2)}{\psi(n-K,n-2)}\Big[%
\sum_{k=2}^{K-1}(-1)^k (k-1)! \psi(n-K,n-k-1)v_k + (-1)^K (K-1)! v_K \Big]
\end{equation*}
\end{proof}

\begin{lemma}
\label{lem:ir2} Suppose $(f,\mathbf{p})$ is a satisfactory mechanism, where $%
f$ is a two-step ranking allocation rule defined by $(\pi_1,\ell)$. Then,
for every $\mathbf{v}$ with $|N^0_{\mathbf{v}}|=n-K$, $K \ge \ell+1$, and $%
v_1 > \ldots > v_K > 0$, we have for every $i \in N^0_{\mathbf{v}}$, 
\begin{align*}
p_i(\mathbf{v}) &= -\frac{(\pi_1-\pi_2)}{\psi(n-\ell,n-2)} \Big[%
\sum_{k=2}^{\ell-1} (-1)^k (k-1)! \psi(n-\ell,n-k-1)v_k + (-1)^{\ell} (\ell
- 1)!v_{\ell} \Big].
\end{align*}
\end{lemma}

\begin{proof}
We follow a similar line of proof as Lemma \ref{lem:ir1}. Consider a
valuation profile $\mathbf{v}$ with $|N^0_{\mathbf{v}}|=n-K$, $K \ge \ell+1$%
, $v_1 > \ldots > v_K > 0$ and $v_j =0$ for all $j > K$.

We now modify Equation \ref{eq:ire1} by using $R^f(\mathbf{v}%
)=(\pi_1-\pi_2)v_2 + \ell \pi_2 v_{\ell+1}$ (by Lemma \ref{lem:ir0}) as
follows: 
\begin{align}  \label{eq:ire11}
0 &= (n-K)p_i(\mathbf{v})+ \sum_{j=1}^K p_j(0,v_{-j}) + (\pi_1-\pi_2)v_2 +
\ell \pi_2 v_{\ell+1}.
\end{align}

Now, for every $j \in \{1,\ldots,K\}$, the profile $\mathbf{v}^j$ has one
more zero-valued agent than the profile $\mathbf{v}$, and hence, we can
apply our induction argument - the base case of $K=\ell$ is solved in Lemma %
\ref{lem:ir1} where we computed $p_i(\mathbf{v})$ with $K \le \ell$ agents
having non-zero valuations. Using induction hypothesis, we simplify Equation %
\ref{eq:ire11} as follows: 
\begin{align*}
-(n-K)p_i(\mathbf{v}) &= \sum_{j=1}^K -\frac{(\pi_1-\pi_2)}{\psi(n-\ell,n-2)}
\Big[\sum_{k=2}^{\ell-1} (-1)^k (k-1)! \psi(n-\ell,n-k-1)v^j_{(k)} +
(-1)^{\ell} (\ell - 1)!v^j_{(\ell)} \Big] \\
&+ (\pi_1-\pi_2)v_2 + \ell \pi_2 v_{\ell+1}.
\end{align*}
This can be rewritten as follows: 
\begin{align}  \label{eq:iir20}
\frac{(n-K)\psi(n-\ell,n-2)}{\pi_1-\pi_2}p_i(\mathbf{v}) &= \sum_{j=1}^K %
\Big[\sum_{k=2}^{\ell-1} (-1)^k (k-1)! \psi(n-\ell,n-k-1)v^j_{(k)} +
(-1)^{\ell} (\ell - 1)!v^j_{(\ell)} \Big]  \notag \\
&- \psi(n-\ell,n-2)v_2 - \frac{\ell \pi_2 \psi(n-\ell,n-2)}{\pi_1-\pi_2}
v_{\ell+1}.
\end{align}
By Proposition \ref{prop:2rank}, 
\begin{align}
\pi_1-\pi_2 &= 1- \frac{(\ell-1)}{C(n-2,\ell-1)+\ell} - \frac{1}{%
C(n-2,\ell-1)+\ell}  \notag \\
&= \frac{C(n-2,\ell-1)}{C(n-2,\ell-1)+\ell}  \notag \\
&= C(n-2,\ell-1)\pi_2  \notag \\
&= \frac{\psi(n-\ell,n-2)}{(\ell-1)!} \pi_2.  \label{eq:reduce}
\end{align}
Hence, Equation \ref{eq:iir20} can be rewritten as 
\begin{align}  \label{eq:iir2}
\frac{(n-K)\psi(n-\ell,n-2)}{\pi_1-\pi_2}p_i(\mathbf{v}) &= \sum_{j=1}^K %
\Big[\sum_{k=2}^{\ell-1} (-1)^k (k-1)! \psi(n-\ell,n-k-1)v^j_{(k)} +
(-1)^{\ell} (\ell - 1)!v^j_{(\ell)} \Big]  \notag \\
&- \psi(n-\ell,n-2)v_2 - \ell ! v_{\ell+1}
\end{align}

Like in Lemma \ref{lem:ir1}, we will rewrite the RHS of Equation \ref%
{eq:iir2} in terms of $v_1,\ldots,v_K$. For this, observe that for any $k$, $%
v_k$ will appear on the RHS of Equation \ref{eq:iir2} if there is some $j
\in \{1,\ldots,K\}$ and some $k^{\prime }\in \{2,\ldots,\ell\}$ such that $%
v^j_{(k^{\prime })}=v_k$. Hence, $v_1$ and $v_{\ell+2},\ldots,v_n$ do not
appear on the RHS of Equation \ref{eq:iir2}. We compute the coefficients of $%
v_k$ for $k \in \{2,\ldots,\ell+1\}$. We consider three cases. \newline

\noindent \textsc{Case 1.} For $v_2$, we note that $v_2 = v^j_{(2)}$ for all 
$j \ne \{1,2\}$. Hence, it has a rank $2$ in $(K-2)$ marginal profiles, and
in each such case, its coefficient in the first summation is 
\begin{equation*}
(-1)^2 (1)! \psi(n-\ell,n-3).
\end{equation*}
Adding this with $-\psi(n-\ell,n-2)$, we get the coefficient of $v_2$ in the
RHS of Equation \ref{eq:iir2} as 
\begin{align*}
&(K-2)\psi(n-\ell,n-3) - \psi(n-\ell,n-2) = -\psi(n-\ell,n-3)(n-K) \\
&= - (-1)^2 (1!)\psi(n-\ell,n-3)(n-K).
\end{align*}

\noindent \textsc{Case 2.} Now, consider $2 < k < \ell$. For $v_{k}$, note
that $v_{k} = v^j_{(k^{\prime })}$ where $k^{\prime }\in \{k,k-1\}$. In
particular, $k^{\prime }=k$ if $j \in \{k+1,\ldots,K\}$ and $k^{\prime }=k-1$
if $j \in \{1,\ldots,k-1\}$. Hence, it has rank $k$ in $(K-k)$ marginal
profiles and rank $(k-1)$ in $(k-1)$ marginal profiles. In the RHS of
Equation \ref{eq:iir2}, the coefficient of $v_{k}$ is $(-1)^{k-1} (k-2)!
\psi(n-\ell,n-k)$ if its rank is $k-1$ and the coefficient is $(-1)^k (k-1)!
\psi(n-\ell,n-k-1)$ if its rank is $k$. Adding them, we get the coefficient
of $v_{k}$ in the RHS of Equation \ref{eq:iir2} as 
\begin{align*}
& (-1)^k (K-k) (k-1)! \psi(n-\ell,n-k-1) + (-1)^{k-1} (k-1) (k-2)!
\psi(n-\ell,n-k) \\
&= (-1)^k (k-1)! \psi(n-\ell,n-k-1) \Big( (K-k) - (n-k) \Big) \\
&= - (-1)^k (n-K) (k-1)! \psi(n-\ell,n-k-1).
\end{align*}

\noindent \textsc{Case 3.} For $v_{\ell}$, note that $v_{\ell} =
v^j_{(k^{\prime })}$ where $k^{\prime }\in \{\ell,\ell-1\}$. In particular, $%
k^{\prime }=\ell$ if $j \in \{\ell+1,\ldots,K\}$ and $k^{\prime }=\ell-1$ if 
$j \in \{1,\ldots,\ell-1\}$. Hence, it has rank $\ell$ in $(K-\ell)$
marginal profiles and rank $(\ell-1)$ in $(\ell-1)$ marginal profiles. In
the RHS of Equation \ref{eq:iir2}, the coefficient of $v_{\ell}$ is $%
(-1)^{\ell-1}(\ell-2)! \psi(n-\ell,n-\ell)$ if its rank is $\ell-1$ and the
coefficient is $(-1)^{\ell} (\ell-1)!$ if its rank is $\ell$. Adding them,
we get the coefficient of $v_{\ell}$ in the RHS of Equation \ref{eq:iir2} as 
\begin{align*}
& (-1)^{\ell-1}(\ell-1) (\ell-2)! \psi(n-\ell,n-\ell) + (-1)^{\ell} (K-\ell)
(\ell-1)! \\
&= (-1)^{\ell} (\ell-1)! \big( (K-\ell) - (n-\ell) \big) \\
&= - (-1)^{\ell} (n-K) (\ell-1)!
\end{align*}

\noindent \textsc{Case 4.} Now, consider $k=\ell+1$. Note that $%
v_{\ell+1}=v^j_{(k^{\prime })}$ if $k^{\prime }=\ell$ and $j \in
\{1,\ldots,\ell\}$. Hence, it has a rank $\ell$ in $\ell$ marginal
economies, where its coefficient in the summation of the RHS of Equation \ref%
{eq:iir2} is 
\begin{equation*}
(-1)^{\ell} (\ell-1)!=(\ell-1)!,
\end{equation*}
since $\ell$ is even. Hence, the coefficient of $v_{\ell+1}$ in the RHS of
Equation \ref{eq:iir2} is $\ell (\ell-1)! - \ell!=0$. \newline

\noindent Aggregating the findings from all the four cases, we can rewrite
Equation \ref{eq:iir2} as 
\begin{align}  \label{eq:iir2}
\frac{(n-K)\psi(n-\ell,n-2)}{\pi_1-\pi_2}p_i(\mathbf{v}) &=
-\sum_{k=1}^{\ell-1}(-1)^k (n-K)(k-1)!\psi(n-\ell,n-k-1) -
(-1)^{\ell}(n-K)(\ell-1)!
\end{align}
This simplifies to the desired expression: 
\begin{equation*}
p_i(\mathbf{v}) = -\frac{(\pi_1-\pi_2)}{\psi(n-\ell,n-2)}\Big[%
\sum_{k=2}^{\ell-1}(-1)^k (k-1)! \psi(n-\ell,n-k-1)v_k + (-1)^{\ell}
(\ell-1)! v_{\ell}\Big]
\end{equation*}
\end{proof}

With the help of these two lemmas, we can now present the proof of Theorem %
\ref{theo:ir}. \newline

\noindent \textsc{Proof of Theorem \ref{theo:ir}.} \newline

\begin{proof}
Consider a two-step allocation rule $(\pi_1,\ell)$ such that $2\ell \le n+1$%
. Proposition \ref{prop:2rank} characterizes the two-step allocation rules
that are satisfactorily implementable. If $\mathbf{p}$ is such that $(f,%
\mathbf{p})$ is a satisfactory mechanism, then it is ex-post individually rational
(by Fact \ref{fact:ir}) if and only if for every $i \in N$ and for every $%
\mathbf{v}$, we have $p_i(0,v_{-i}) \le 0$.

Fix $i \in N$ and choose a profile $(0,v_{-i})$. By Lemma \ref{lem:contin}, $%
R^f$ is continuous in $\mathbf{v}$. Hence, by the expression of $%
p_i(0,v_{-i})$ in Theorem \ref{theo:char}, $p_i(0,v_{-i})$ is continuous in $%
v_{-i}$. Hence, we only consider $v_{-i}$ such that $(0,v_{-i})$ is $0$%
-generic. Thus, we can apply Lemma \ref{lem:ir1} and \ref{lem:ir2} to
compute $p_i(0,v_{-i})$ and show that it is non-positive.

Suppose $v_1 > v_2 > \ldots > v_K > 0$ and $v_j=0$ for all $j > K$. By
Lemmas \ref{lem:ir1} and \ref{lem:ir2}, $p_i(0,v_{-i}) \le 0$ if and only if
for every $K \le \ell$, the following summation is non-negative: 
\begin{align*}
\Big[\sum_{j=2}^{K-1} (-1)^j (j-1)! \psi(n-K,n-j-1)v_j + (-1)^K (K-1)!v_K %
\Big].
\end{align*}
Expanding this, we get 
\begin{align}  \label{e:tir}
1! \psi(n-K,n-3)v_2 - 2!\psi(n-K,n-4)v_3 + \ldots + (-1)^K
(K-1)!\psi(n-K,n-K-1)v_K,
\end{align}
where we abused notation to define $\psi(n-K,n-K-1) \equiv 1$. Note that if $%
K$ is even the last term of Expression \ref{e:tir} is positive. In that
case, it is sufficient to show that this summation is non-negative till $K-1$
(i.e., the last negative term in the expression). This idea is captured by
considering the summation till $\lfloor K \rfloor_o$ (the largest odd number
less than or equal to $K$). Hence, we need to show the following expression
is non-negative: 
\begin{align*}
& \sum_{j=2}^{\lfloor K \rfloor_o} (-1)^j (j-1)! \psi(n-K,n-j-1)v_j \\
&= \sum_{2 \le j \le \lfloor K \rfloor_o: j~\text{even}}\Big[ %
(j-1)!\psi(n-K,n-j-1)v_j - (j!) \psi(n-K,n-j-2)v_{j+1} \Big] \\
&= \sum_{2 \le j \le \lfloor K \rfloor_o: j~\text{even}}(j-1)!
\psi(n-K,n-j-2) \Big[ (n-j-1)v_j - jv_{j+1} \Big]. \\
&\ge \sum_{2 \le j \le \lfloor K \rfloor_o: j~\text{even}}(j-1)!
\psi(n-K,n-j-2) (n-2j-1)v_j.
\end{align*}
Note that we are consider a 2-step allocation rule $(\pi,\ell)$ such that $%
2\ell \le n+1$. Since $K \le \ell$, for every $2 \le j \le \lfloor K
\rfloor_o: j~\text{even}$, we have $j+1 \le \ell$. Hence, for every $2 \le j
\le \lfloor K \rfloor_o: j~\text{even}$, we have $2(j+1) \le n+1$ or $n-2j-1
\ge 0$. This implies that the above expression is non-negative, which
completes the proof.
\end{proof}

\subsection*{Proof of Proposition \protect\ref{prop:2steppay}}

\begin{proof}
Consider a valuation profile $\mathbf{v}$ with $v_1 > v_2 > \ldots > v_n > 0$%
. By Proposition \ref{prop:2rank}, 
\begin{align}
\pi_1-\pi_2 &= 1- \frac{(\ell-1)}{C(n-2,\ell-1)+\ell} - \frac{1}{%
C(n-2,\ell-1)+\ell}  \notag \\
&= \frac{C(n-2,\ell-1)}{C(n-2,\ell-1)+\ell}  \notag \\
&= C(n-2,\ell-1)\pi_2  \notag \\
&= \frac{\psi(n-\ell,n-2)}{(\ell-1)!} \pi_2.  \label{eq:reduce}
\end{align}
Then, the payments are computed using Lemma \ref{lem:ir2} as follows. 
\begin{align*}
p_1(\mathbf{v}) &= p_1(0,v_{-1}) + v_1 \pi_1 -
\int_0^{v_1}f_1(x_1,v_{-1})dx_1 \\
&= p_1(0,v_{-1}) + v_1 \pi_1 - (v_1-v_2)\pi_1 - (v_2 - v_{\ell+1})\pi_2 \\
&= p_1(0,v_{-1}) + v_2 (\pi_1 - \pi_2) + v_{\ell+1} \pi_2 \\
&= -\frac{\pi_2}{(\ell-1)!} \Big[ \sum_{k=2}^{\ell-1}(-1)^k (k-1)!
\psi(n-\ell,n-k-1)v_{k+1} \Big] - v_{\ell+1}\pi_2 + v_2 (\pi_1 - \pi_2) +
v_{\ell+1} \pi_2 \\
& \text{(The above simplification uses Lemma \ref{lem:ir2} along with
Equation \ref{eq:reduce} and the fact that $\ell$ is even.)} \\
&= -\frac{\pi_2}{(\ell-1)!} \Big[ \sum_{k=2}^{\ell-1}(-1)^k (k-1)!
\psi(n-\ell,n-k-1)v_{k+1} \Big] + \frac{\psi(n-\ell,n-2)}{(\ell-1)!} v_2
\pi_2 \\
&= -\frac{\pi_2}{(\ell-1)!} \Big[ \sum_{k=1}^{\ell-1}(-1)^k (k-1)!
\psi(n-\ell,n-k-1)v_{k+1} \Big] \\
\end{align*}

For every $i \in \{2,\ldots,\ell\}$,

\begin{align*}
p_i(\mathbf{v}) &= p_i(0,v_{-i}) + v_i \pi_2 -
\int_0^{v_i}f_i(x_i,v_{-i})dx_i \\
&= p_i(0,v_{-i}) + v_i \pi_2 - (v_i - v_{\ell+1})\pi_2 \\
&= p_i(0,v_{-i}) + v_{\ell+1}\pi_2 \\
&= -\frac{\pi_2}{(\ell-1)!} \Big[ \sum_{k=2}^{i-1}(-1)^k (k-1)!
\psi(n-\ell,n-k-1)v_{k} + \sum_{k=i}^{\ell-1}(-1)^k (k-1)!
\psi(n-\ell,n-k-1) v_{k+1} \Big] \\
&- v_{\ell+1}\pi_2 + v_{\ell+1}\pi_2 \\
& \text{(The above simplification uses Lemma \ref{lem:ir2} along with
Equation \ref{eq:reduce} and the fact that $\ell$ is even.)} \\
&= - \frac{\pi_2}{(\ell-1)!} \Big[ \sum_{k=2}^{i-1}(-1)^k (k-1)!
\psi(n-\ell,n-k-1)v_{k} + \sum_{k=i}^{\ell-1}(-1)^k (k-1)!
\psi(n-\ell,n-k-1) v_{k+1} \Big]
\end{align*}

For every $i > \ell$, we directly use Lemma \ref{lem:ir2} along with
Equation (\ref{eq:reduce}) to get 
\begin{align*}
p_i(\mathbf{v}) &= p_i(0,v_{-i}) = -\frac{\pi_2}{(\ell-1)!} \Big[ %
\sum_{k=2}^{\ell-1}(-1)^k (k-1)! \psi(n-\ell,n-k-1)v_{k} + (-1)^{\ell}
(\ell-1)!v_{\ell}\Big]
\end{align*}
\end{proof}

\subsection*{Proof of Theorem \protect\ref{theo:rgl} and Proposition \protect
\ref{prop:rpar}}

In this section, we give proofs of Theorem \ref{theo:rgl} and Proposition %
\ref{prop:rpar}. We first show that every r-Pareto optimal allocation rule
satisfies the fact that probabilities add up to $1$, i.e., the good is never
wasted.

\begin{lemma}
\label{lem:add1} If $f$ is an r-Pareto optimal or an r-optimal ranking
allocation rule with probabilities $(\pi_1,\ldots,\pi_n)$, then 
\begin{equation*}
\sum_{i \in N}\pi_i=1.
\end{equation*}
\end{lemma}

\begin{proof}
Suppose $f$ is a ranking allocation rule with probabilities $%
(\pi_1,\ldots,\pi_n)$. Assume for contradiction $f$ is r-optimal but $%
\sum_{i=1}^n \pi_i < 1$. Let $\delta = 1- \sum_{i \in N}\pi_i > 0$. We
construct another ranking allocation rule $f^{\prime }$ with probabilities $%
\pi^{\prime }_i \equiv \pi_i + \frac{\delta}{n}$ for all $i \in N$. Note
that $\sum_{i \in N}\pi^{\prime }_i = 1$ and 
\begin{align*}
\sum_{k=1}^n (-1)^k C(n-1,k-1)\pi^{\prime }_k &= \sum_{k=1}^n (-1)^k
C(n-1,k-1)\pi_k + \frac{\delta}{n}\sum_{k=1}^n (-1)^k C(n-1,k-1) \\
&= \sum_{k=1}^n (-1)^k C(n-1,k-1)\pi_k \\
&= 0,
\end{align*}
where the first equality is from the definition of $(\pi^{\prime
}_1,\ldots,\pi^{\prime }_n)$, the second equality follows from the fact that 
$\sum_{k=1}^n (-1)^k C(n-1,k-1)=0$, and the third equality follows from
Proposition \ref{prop:rank} and the fact that $(\pi_1,\ldots,\pi_n)$ is a
satisfactorily implementable ranking allocation rule. Hence, by Proposition %
\ref{prop:rank}, $f^{\prime }$ is satisfactorily implementable. But this
contradicts the r-optimality of $f$.

Now, suppose $f$ is r-Pareto optimal. The above argument also implies that
at every valuation profile $\mathbf{v}$, we have 
\begin{equation*}
\sum_{i \in N}v_if^{\prime }_i(\mathbf{v}) \ge \sum_{i \in N}v_if_i(\mathbf{v%
}),
\end{equation*}
with strict inequality holding at almost everywhere. This contradicts the
fact that $f$ is r-Pareto optimal.
\end{proof}

This leads to a simplification of r-Pareto optimality in terms of
first-order stochastic-dominance.

\begin{defn}
A ranking allocation rule $f$ with probabilities $(\pi_1,\ldots,\pi_n)$ 
\textbf{first-order stochastic-dominates (FOSD)} a ranking allocation rule $%
f^{\prime }$ with probabilities $(\pi^{\prime }_1,\ldots,\pi^{\prime }_n)$
if for every $j \in N$, we have 
\begin{equation*}
\sum_{i \le j}\pi_i \ge \sum_{i \le j}\pi^{\prime }_i,
\end{equation*}
with strict inequality holding at least once. In this case, we write $f
\succ_{FOSD} f^{\prime }$.
\end{defn}

\begin{lemma}
\label{lem:fosd} Suppose $f$ is a ranking allocation rule with probabilities 
$(\pi_1,\ldots,\pi_n)$ such that it is satisfactorily implementable. Then, $%
f $ is r-Pareto optimal if and only if

\begin{enumerate}
\item $\sum_{i \in N}\pi_i = 1$ and

\item if there exists no ranking allocation rule $f^{\prime }$ with
probabilities $(\pi^{\prime }_1,\ldots,\pi^{\prime }_n)$ such that

\begin{itemize}
\item $\sum_{i \in N}\pi^{\prime }_i=1$,

\item $f^{\prime }$ is satisfactorily implementable, and

\item $f^{\prime }\succ_{FOSD} f$.
\end{itemize}
\end{enumerate}
\end{lemma}

\begin{proof}
Suppose a ranking rule $f$ with probabilities $(\pi_1,\ldots,\pi_n)$ is
r-Pareto optimal. By Lemma \ref{lem:add1}, we know that $\sum_{i \in N}\pi_i
= 1$. Now, assume for contradiction that there exists a ranking allocation
rule $f^{\prime }$ with probabilities $(\pi^{\prime }_1,\ldots,\pi^{\prime
}_n)$ such that $f$' is satisfactorily implementable, $\sum_{i \in
N}\pi^{\prime }_i=1$, and $f^{\prime }\succ_{FOSD} f$. Since $f^{\prime
}\succ_{FOSD} f$, for any profile of generic valuations $\mathbf{v}$ with $%
v_1 > v_2 > \ldots > v_n$, we have 
\begin{equation*}
\sum_{i \in N}v_i \pi^{\prime }_i \ge \sum_{i \in N}v_i \pi_i.
\end{equation*}
The strict inequality must hold for some generic valuation profile by the
definition of first-order stochastic dominance. Now, take any arbitrary
valuation profile $\mathbf{v}$. Note that the total welfare of a ranking
allocation rule is continuous in the valuations of the agents. Hence, it can
be written as a limit point of generic valuation profiles like above. This
implies that for every valuation profile $\mathbf{v}$, we have 
\begin{equation*}
\sum_{i \in N}v_i f^{\prime }_i(\mathbf{v}) \ge \sum_{i \in N}v_i f_i(%
\mathbf{v}),
\end{equation*}
with strict inequality holding for some $\mathbf{v}$. This implies that $f$
is not r-Pareto optimal, a contradiction.

Now, for the other direction suppose $f$ is a ranking allocation with
probabilities $(\pi_1,\ldots,\pi_n)$ satisfying the properties in the claim.
Assume for contradiction that $f$ is not Pareto optimal. Then, there exists
a satisfactorily implementable ranking allocation rule $f^{\prime }$ with
probabilities $(\pi^{\prime }_1,\ldots,\pi^{\prime }_n)$ such that for
valuation profiles $\mathbf{v}$, we have 
\begin{equation*}
\sum_{i \in N}v_i f^{\prime }_i(\mathbf{v}) \ge \sum_{i \in N}v_i f_i(%
\mathbf{v}),
\end{equation*}
with strict inequality holding for some $\mathbf{v}$. By Lemma \ref{lem:add1}%
, we can assume $\sum_{i \in N}\pi^{\prime }_i=1$ without loss of
generality. For generic valuation profiles $\mathbf{v}$ with $v_1 > \ldots >
v_n$, we have $\sum_{i \in N}v_i \pi^{\prime }_i \ge \sum_{i \in N}v_i
\pi_i. $ As in the previous paragraph, continuity of the total welfare of
agents in a ranking allocation rule implies that $f^{\prime }\succ_{FOSD} f$%
. This is a contradiction.
\end{proof}

We now provide a proof of Theorem \ref{theo:rgl}. \newline

\noindent \textsc{Proof of Theorem \ref{theo:rgl}.} \newline

\begin{proof}
We denote the GL allocation rule as $f^G$. Assume for contradiction that $%
f^G $ is not r-Pareto optimal. By Lemma \ref{lem:fosd}, there is another
ranking allocation rule $f$ such that $f$ is satisfactorily implementable
and $f \succ_{FOSD} f^G$. Suppose the allocation probabilities of $f$ are $%
(\pi_1,\ldots,\pi_n)$. We know that the allocation probabilities of $f^G$
are $(1-1/n, 1/n, 0, 0, \ldots,0)$. Since $f \succ_{FOSD} f^G$, $%
\pi_1+\pi_2=1$, and hence, $\pi_3=\ldots=\pi_n=0$. Since $f$ is
satisfactorily implementable, by Proposition \ref{prop:rank}, we get 
\begin{equation*}
\pi_1 - (n-1)\pi_2 = 0.
\end{equation*}
Using $\pi_1+\pi_2=1$ and simplifying, we get $\pi_1=1-1/n$. Hence, $f$ is
the Green-Laffont allocation rule, which is a contradiction.

The above proof along with Lemma \ref{lem:add1} also makes it clear that
among all ranking allocation rules which allocates probability to only $%
\pi_1 $ and $\pi_2$, the GL allocation rule is the unique r-Pareto optimal
allocation rule.
\end{proof}

We now provide a proof of Proposition \ref{prop:rpar}. \newline

\noindent Proof of Proposition \ref{prop:rpar}. \newline

\begin{proof}
Suppose $n \le 8$. Then, the GL allocation rule is an r-optimal allocation
rule by Corollary \ref{cor:comp}. Since $\pi_1+\pi_2=1$ in the GL allocation
rule, this implies that the GL allocation rule dominates every other
satisfactorily implementable ranking allocation rule an FOSD sense. By Lemma %
\ref{lem:fosd}, the GL allocation rule is the unique r-Pareto optimal
allocation rule.

Suppose $n > 8$. Then, Theorem \ref{theo:nrgl} implies that there is a
unique r-optimal allocation rule. Hence, no other satisfactorily
implementable ranking allocation rule can dominate this unique r-optimal
allocation rule in an FOSD sense. By Lemma \ref{lem:fosd}, this unique
r-optimal allocation rule is then r-Pareto optimal.

Finally choose an r-Pareto optimal allocation rule $(\pi_1,\ldots,\pi_n)$.
By definition of $\pi_1^*$, we have $\pi_1 \le \pi_1^*$. Further, if $\pi_1
< 1-1/n$, the GL allocation rule dominates this allocation rule in an FOSD
sense, and by Lemma \ref{lem:fosd}, it is not r-Pareto optimal. Hence, $%
\pi_1 \ge 1-1/n$.
\end{proof}


\end{document}